\documentclass{vldb}

\usepackage{graphicx}
\usepackage{balance}
\usepackage{hyperref}
\usepackage{xspace}
\usepackage{amsmath,amssymb,mathtools}
\usepackage{array,booktabs,multirow}
\usepackage{enumitem}
\usepackage{algorithm}
\usepackage[noend]{algpseudocode}
\usepackage{listings}
\usepackage{subcaption}
\usepackage{microtype}
\usepackage{titlesec}
\usepackage{tikz}
\usetikzlibrary{arrows.meta,positioning}

\newtheorem{theorem}{Theorem}[section]
\newtheorem{lemma}[theorem]{Lemma}
\newtheorem{corollary}[theorem]{Corollary}
\newtheorem{proposition}[theorem]{Proposition}
\newtheorem{definition}[theorem]{Definition}
\newtheorem{assumption}[theorem]{Assumption}

\hypersetup{hidelinks}
\titleformat{\section}{\bfseries\large}{\thesection.}{0.5em}{\MakeUppercase}
\titleformat{\subsection}{\bfseries\normalsize}{\thesubsection}{0.5em}{}
\titleformat{\subsubsection}{\bfseries\normalsize}{\thesubsubsection}{0.5em}{}
\titlespacing*{\section}{0pt}{1.8ex plus .5ex minus .2ex}{0.8ex}
\titlespacing*{\subsection}{0pt}{2.0ex plus .5ex minus .2ex}{0.8ex}
\titlespacing*{\subsubsection}{0pt}{1.2ex plus .3ex minus .2ex}{0.5ex}
\newcommand\nop[1]{}
\newcommand{\TrieGS}{\textsc{TrieGS}\xspace}
\newcommand{\LFNT}{\textsc{LFNT}\xspace}
\newcommand{\OQ}{\mathcal{O}_{Q}}
\newcommand{\RootVec}{\mathsf{RootVec}}
\newcommand{\pred}{\mathsf{pred}}
\newcommand{\mult}{\mathsf{mult}}
\newcommand{\getM}{\mathsf{getM}}
\newcommand{\setM}{\mathsf{setM}}
\newcommand{\scanIndex}{\mathsf{scanIndex}}

\newcommand{\Vars}{\mathsf{Vars}}

\newcommand{\oldv}{\mathsf{old}}
\newcommand{\newv}{\mathsf{new}}

\newcommand{\idx}[1]{I_{#1}}

\tikzset{
  triegs/entity/.style={draw,rounded corners=2pt,minimum width=1.25cm,
    minimum height=0.55cm,align=center,fill=gray!8,font=\scriptsize},
  triegs/event/.style={draw,rounded corners=2pt,minimum width=1.55cm,
    minimum height=0.62cm,align=center,fill=blue!7,font=\scriptsize},
  triegs/view/.style={draw,rounded corners=2pt,minimum width=1.72cm,
    minimum height=0.58cm,align=center,fill=blue!8,font=\scriptsize},
  triegs/proj/.style={draw,dashed,rounded corners=2pt,minimum width=1.72cm,
    minimum height=0.58cm,align=center,fill=orange!10,font=\scriptsize},
  triegs/change/.style={draw,rounded corners=2pt,minimum width=1.78cm,
    minimum height=0.58cm,align=center,fill=green!10,font=\scriptsize},
  triegs/root/.style={draw,very thick,rounded corners=2pt,minimum width=1.9cm,
    minimum height=0.62cm,align=center,fill=green!12,font=\scriptsize},
  triegs/edge/.style={-{Latex[length=2mm]},line width=0.55pt},
  triegs/update/.style={-{Latex[length=2mm]},line width=0.8pt,draw=green!50!black},
  triegs/old/.style={draw,rounded corners=2pt,minimum width=1.42cm,
    minimum height=0.52cm,align=center,fill=gray!10,font=\scriptsize},
  triegs/new/.style={draw,rounded corners=2pt,minimum width=1.42cm,
    minimum height=0.52cm,align=center,fill=green!10,font=\scriptsize}
}

\newcommand{\TrieGSRunningOverview}{%
\begin{tikzpicture}[x=1cm,y=1cm,font=\scriptsize]
  \node[font=\bfseries\scriptsize] at (2.2,2.2) {(a) Query instance at $t_0$};
  \node[triegs/entity] (tweet1) at (2.0,0.75) {\texttt{tweet1}};
  \node[triegs/entity] (max) at (0.45,1.8) {\texttt{Max}};
  \node[triegs/entity] (tweet2) at (4.2,0.75) {\texttt{tweet2}};
  \node[triegs/entity] (eva) at (5.75,1.8) {\texttt{Eva}};
  \node[triegs/entity] (mia) at (4.2,-0.75) {\texttt{Mia}};

  \draw[triegs/edge] (tweet1) -- node[midway,below=2pt,fill=white,inner sep=1pt]
    {\texttt{hasCreator}} (max);
  \draw[triegs/edge] (tweet1) -- node[above=2pt] {\texttt{hasReply} $\times 2$} (tweet2);
  \draw[triegs/edge] (tweet2) -- node[midway,below=2pt,fill=white,inner sep=1pt]
    {\texttt{hasCreator}} (eva);
  \draw[triegs/update] (tweet2) -- node[right=2pt,align=left]
    {\texttt{likedBy}\\$0\!\rightarrow\!1$ at $t_1$} (mia);
  \node[align=center,font=\scriptsize] at (2.55,-1.55)
    {$t_0$: no answer\qquad $t_1$: one weighted answer, count $2$};

  \draw[gray!50] (6.65,-1.85) -- (6.65,2.4);

  \node[font=\bfseries\scriptsize] at (10.15,2.2) {(b) Stream and snapshot evolution};
  \draw[-{Latex[length=2mm]},line width=0.65pt] (7.4,1.35) -- (12.85,1.35);
  \foreach \x/\lab in {7.85/$t_0$,9.7/$t_1$,11.65/$t_2$} {
    \draw (\x,1.26)--(\x,1.44);
    \node[above=1pt] at (\x,1.44) {\lab};
  }
  \node[align=center] at (7.85,0.72) {initial\\count $0$};
  \node[align=center] at (9.7,0.72) {$+\,\texttt{likedBy}$\\count $2$};
  \node[align=center] at (11.65,0.72) {$-\,\texttt{hasReply}$\\count $1$};

  \node[anchor=east,font=\scriptsize] at (7.35,-0.12) {Updater};
  \draw[triegs/edge] (7.55,-0.12) -- node[above=3pt] {build $L_{t_1}$} (9.3,-0.12);
  \node[circle,fill=green!50!black,inner sep=1.6pt] at (9.55,-0.12) {};
  \node[below=4pt] at (9.55,-0.12) {publish $S_{t_1}$};
  \draw[triegs/edge] (9.8,-0.12) -- node[above=3pt] {build $L_{t_2}$} (11.2,-0.12);
  \node[circle,fill=green!50!black,inner sep=1.6pt] at (11.45,-0.12) {};
  \node[below=4pt] at (11.45,-0.12) {publish $S_{t_2}$};

  \node[anchor=east,font=\scriptsize] at (7.35,-1.12) {Enumerator $E_1$};
  \draw[line width=1.3pt,draw=blue!55] (9.55,-1.12) -- (12.45,-1.12);
  \node[below=2pt,align=center] at (11.0,-1.12)
    {pins $S_{t_1}$ and keeps count $2$\\while $t_2$ is built and published};
\end{tikzpicture}%
}

\newcommand{\TrieGSViewEvolution}{%
\begin{tikzpicture}[x=1cm,y=1cm,font=\scriptsize]
  \node[font=\bfseries\scriptsize] at (3.15,3.25) {(a) Populated logical views for $z=\texttt{tweet2}$};
  \node[triegs/root] (v0) at (3.15,2.55) {$V_0(z)$\\$0\!\rightarrow\!2$};
  \node[triegs/proj] (v1p) at (1.25,1.55) {$V_1^p(z)$\\$2$};
  \node[triegs/proj] (v2p) at (5.05,1.55) {$V_2^p(z)$\\$0\!\rightarrow\!1$};
  \node[triegs/view] (v1) at (1.25,0.45) {$V_1(v,z)$\\$2$};
  \node[triegs/change] (v2) at (5.05,0.45) {$V_2(z,y)$\\$0\!\rightarrow\!1$};
  \node[triegs/old] (leftbase) at (1.25,-0.72)
    {\texttt{hasReply}$=2$\\$H^p(v)=1$};
  \node[triegs/change] (rightbase) at (5.05,-0.72)
    {\texttt{hasCreator}$=1$\\$L^p(z):0\!\rightarrow\!1$};

  \draw[triegs/edge] (v1p) -- (v0);
  \draw[triegs/edge] (v2p) -- (v0);
  \draw[triegs/edge] (v1) -- (v1p);
  \draw[triegs/edge] (v2) -- (v2p);
  \draw[triegs/edge] (leftbase) -- (v1);
  \draw[triegs/edge] (rightbase) -- (v2);
  \node[align=center] at (3.15,-1.55)
    {$V_1^p(\texttt{tweet2})=2$, $V_2^p(\texttt{tweet2})=1$,\\
     hence $V_0(\texttt{tweet2})=2\cdot 1=2$.};

  \draw[gray!50] (6.45,-1.8) -- (6.45,3.45);

  \node[font=\bfseries\scriptsize] at (9.55,3.25) {(b) Exact delta propagation at $t_1$};
  \node[triegs/change] (b0) at (8.85,2.55) {\texttt{likedBy(tweet2,Mia)}\\$0\!\rightarrow\!1$, $\delta=+1$};
  \node[triegs/change] (b1) at (8.85,1.62) {$L^p(\texttt{tweet2})$\\$0\!\rightarrow\!1$, $\delta=+1$};
  \node[triegs/change] (b2) at (8.85,0.69) {$V_2(\texttt{tweet2},\texttt{Eva})$\\$0\!\rightarrow\!1$, $\delta=+1$};
  \node[triegs/change] (b3) at (8.85,-0.24) {$V_2^p(\texttt{tweet2})$\\$0\!\rightarrow\!1$, $\delta=+1$};
  \node[triegs/root] (b4) at (8.85,-1.17) {$V_0(\texttt{tweet2})$\\$0\!\rightarrow\!2$, $\delta=+2$};
  \draw[triegs/update] (b0) -- (b1);
  \draw[triegs/update] (b1) -- (b2);
  \draw[triegs/update] (b2) -- (b3);
  \draw[triegs/update] (b3) -- (b4);
  \node[draw,rounded corners=2pt,fill=yellow!10,text width=2.45cm,
    align=left,font=\scriptsize] (why) at (11.35,-0.65)
    {Why $\delta=+2$? The unchanged reply branch contributes two compatible witnesses.};
  \draw[dashed,-{Latex[length=1.8mm]}] (b4.east) -- (why.west);
\end{tikzpicture}%
}

\newcommand{\TrieGSSnapshotPublication}{%
\begin{tikzpicture}[x=1cm,y=1cm,font=\scriptsize]
  \node[draw,very thick,rounded corners=2pt,fill=yellow!13,
    minimum width=3.4cm,minimum height=0.65cm,align=center] (pub) at (6.2,3.2)
    {$\mathsf{PublishedRootVec}$\\one atomic pointer};

  \node[triegs/root,minimum width=2.25cm] (rv1) at (2.75,2.05)
    {$\RootVec_{t_1}$\\published, pinned};
  \node[triegs/root,minimum width=2.25cm] (rv2) at (9.65,2.05)
    {$\RootVec_{t_2}$\\published later};
  \draw[triegs/edge] (pub.south west) -- node[above left] {load by $E_1$} (rv1.north);
  \draw[triegs/update] (pub.south east) -- node[above right] {atomic publish} (rv2.north);

  \node[triegs/old] (v11) at (1.2,0.8) {$V_1@t_1$\\count $2$};
  \node[triegs/old] (v21) at (2.75,0.8) {$V_2@t_1$\\count $1$};
  \node[triegs/old] (a1) at (4.3,0.8) {$\Ans@t_1$\\count $2$};
  \node[triegs/new] (v12) at (8.1,0.8) {$V_1@t_2$\\count $1$};
  \node[triegs/new] (v22) at (9.65,0.8) {$V_2@t_2$\\count $1$};
  \node[triegs/new] (a2) at (11.2,0.8) {$\Ans@t_2$\\count $1$};

  \draw[triegs/edge] (rv1) -- (v11);
  \draw[triegs/edge] (rv1) -- (v21);
  \draw[triegs/edge] (rv1) -- (a1);
  \draw[triegs/edge] (rv2) -- (v12);
  \draw[triegs/edge] (rv2) -- (v22);
  \draw[triegs/edge] (rv2) -- (a2);

  \node[draw,rounded corners=2pt,fill=blue!8,align=center,
    minimum width=3.25cm,minimum height=0.72cm] (enum) at (2.75,-0.55)
    {Enumerator $E_1$ traverses only $t_1$ roots\\and emits the weighted answer with count $2$};
  \node[draw,rounded corners=2pt,fill=green!8,align=center,
    minimum width=3.25cm,minimum height=0.72cm] (upd) at (9.65,-0.55)
    {Update workers build and publish $t_2$\\without modifying nodes visible from $t_1$};
  \draw[triegs/edge] (a1) -- (enum);
  \draw[triegs/update] (a2) -- (upd);

  \node[align=center,font=\scriptsize] at (6.2,-1.55)
    {The enumerator never assembles a snapshot by independently reading live relation roots;\\
     the vector is the query-level consistency boundary.};
\end{tikzpicture}%
}

\newcommand{\TrieGSLFNTPhysical}{%
\begin{tikzpicture}[x=1cm,y=1cm,font=\scriptsize]
  \node[font=\bfseries\scriptsize] at (3.0,2.25)
    {(a) Logical Gma and access orders};
  \node[draw,rounded corners=2pt,fill=blue!8,minimum width=4.9cm,
    minimum height=0.72cm,align=center] (gma) at (3.0,1.35)
    {$V_1(\mathsf{tweet1},\mathsf{tweet2})=2$};
  \node[draw,rounded corners=2pt,fill=gray!8,minimum width=5.1cm,
    minimum height=0.72cm,align=center] (ord1) at (3.0,0.25)
    {update order: $[V_1\mid\mathsf{tweet1}\mid\mathsf{tweet2}]\rightarrow 2$};
  \node[draw,rounded corners=2pt,fill=orange!9,minimum width=5.1cm,
    minimum height=0.72cm,align=center] (ord2) at (3.0,-0.85)
    {enumeration order: $[V_1\mid\mathsf{tweet2}\mid\mathsf{tweet1}]\rightarrow 2$};
  \draw[triegs/edge] (gma) -- (ord1);
  \draw[triegs/edge] (gma) -- (ord2);

  \draw[gray!50] (6.25,-1.55) -- (6.25,2.5);

  \node[font=\bfseries\scriptsize] at (9.45,2.25)
    {(b) Copy-on-write path replacement};
  \node[triegs/root,minimum width=1.8cm] (r1) at (7.45,1.45)
    {$r^{t_1}_{V_1}$};
  \node[triegs/root,minimum width=1.8cm] (r2) at (11.45,1.45)
    {$r^{t_2}_{V_1}$};
  \node[triegs/old] (n1) at (7.45,0.35)
    {old path node};
  \node[triegs/new] (n2) at (11.45,0.35)
    {replacement path node};
  \node[draw,rounded corners=2pt,fill=yellow!10,minimum width=1.65cm,
    minimum height=0.52cm,align=center,font=\scriptsize] (shared) at (9.45,-0.05)
    {shared unaffected\\subtree};
  \node[triegs/old] (p1) at (7.45,-0.85)
    {old payload\\count $2$};
  \node[triegs/new] (p2) at (11.45,-0.85)
    {new payload\\count $1$};

  \draw[triegs/edge] (r1) -- (n1);
  \draw[triegs/update] (r2) -- (n2);
  \draw[triegs/edge] (n1) -- (p1);
  \draw[triegs/update] (n2) -- (p2);
  \draw[dashed] (n1) -- (shared);
  \draw[dashed] (n2) -- (shared);
  \node[align=center,text width=5.0cm] at (9.45,-1.8)
    {Only nodes on the changed path are replaced; unaffected subtries remain
     shared. The $t_1$ path and payload stay immutable until all readers unpin
     $t_1$.};
\end{tikzpicture}%
}

\MakeRobust{\Call}

\vldbTitle{CQELS-TrieGS Report: Snapshot-Consistent Constant-Delay Enumeration for Streaming Graph Queries}
\vldbAuthors{Danh Le-Phuoc}
\vldbDOI{XX.XX/XXX.XX}
\vldbVolume{XX}
\vldbNumber{XXX}
\vldbYear{2027}

\begin{document}

\title{CQELS-TrieGS Report: Snapshot-Consistent Constant-Delay Enumeration for Streaming Graph Queries}

\numberofauthors{1}
\author{\alignauthor Danh Le-Phuoc \\
Technical University Berlin}
\maketitle

\begin{abstract}
Continuous graph-query engines must process edge updates while making current
query results available to concurrent consumers. Existing dynamic
constant-delay enumeration methods provide strong per-answer delay guarantees,
but are commonly formulated as a maintenance-then-enumeration process.
Conversely, multicore graph-stream engines emphasize update throughput and
match discovery without a query-level snapshot guarantee for concurrent
full-result enumeration.

We present \TrieGS, a shared-memory engine that maintains a query-specific CDE
state over a streaming graph. Logically, \TrieGS uses the classical
free-connex witness-subtree enumerator; dynamically, it maintains multiplicity
payloads using exact signed deltas; physically, RDF terms are dictionary-encoded
and the required access structures are realized as versioned \LFNT relations.
The new systems problem is not relation-level snapshotting alone: an enumeration
job must observe one consistent state across all interdependent base relations,
projections, views, and indexes. \TrieGS therefore publishes an atomic root
vector only after an update epoch has fully propagated. Enumeration threads pin
one published root vector and traverse immutable versions while later updates
continue.

For fixed full acyclic conjunctive graph queries under multiset semantics, we
prove that every enumeration job observes a graph state corresponding to one
published logical time and enumerates its nonzero answer--multiplicity pairs
with $\OQ(1)$ delay. Projected free-connex queries are supported by explicitly
maintaining a count-aware answer relation; its output-size-dependent update and
space costs are included in the bounds. Sliding windows are reduced to arrivals
and expiration deletions. Practical extensions, including SWAG aggregate
summaries, compressed access structures, filters, aggregation, multi-query
sharing, and bounded-width cyclic bags, use the same substrate but have separate
or parameterized guarantees.
\end{abstract}

\section{Introduction}
\label{sec:introduction}

Graph streams arise when social interactions, communication events, mobility
records, or knowledge-graph assertions arrive continuously. Systems such as
GraphFlow, TurboFlux, and later continuous subgraph matching engines maintain
matches as the graph changes~\cite{DBLP:conf/sigmod/KankanamgeSMCS17,
DBLP:conf/sigmod/KimSHLHCSJ18,DBLP:journals/pvldb/SunSHL22}. Their performance
is highly workload dependent: a common-framework study found that candidate
search, index maintenance, and update selectivity dominate in different
regimes, and no single method wins universally~\cite{DBLP:journals/pvldb/SunSLH22}.

A complementary line of work studies dynamic query evaluation through
constant-delay enumeration (CDE). For suitable conjunctive-query classes, a
maintained representation can produce the first answer and each subsequent
answer with delay independent of the data size~\cite{bagan2007acyclic,
berkholz2017answering,DBLP:conf/sigmod/IdrisUV17}. Dynamic Yannakakis and CROWN
show how compact views or join-free change propagation can maintain such
representations under updates~\cite{DBLP:conf/sigmod/IdrisUV17,WangHDY23}.
These results establish the query-theoretic foundation of this paper.
claim acyclic CDE itself as new.

\paragraph*{Running example.}
Figure~\ref{fig:running-overview} introduces the example used throughout the
report. The query reports a message $v$, its creator $x$, a replied-to message
$z$, the creator $y$ of $z$, and a user $w$ who liked $z$. At $t_0$, two active
reply events with distinct event identifiers fall inside the same window and
represent the assertion
$\mathsf{hasReply}(\mathsf{tweet1},\mathsf{tweet2})$ twice, but no matching
$\mathsf{likedBy}$ event exists. Inserting
$\mathsf{likedBy}(\mathsf{tweet2},\mathsf{Mia})$ at $t_1$ therefore creates one
distinct answer
\[
 (\mathsf{Max},\mathsf{Eva},\mathsf{tweet1},\mathsf{tweet2},\mathsf{Mia})
\]
with multiplicity two. At $t_2$, one reply event expires from the active
window, reducing that multiplicity from two to one.

\begin{figure*}[t]
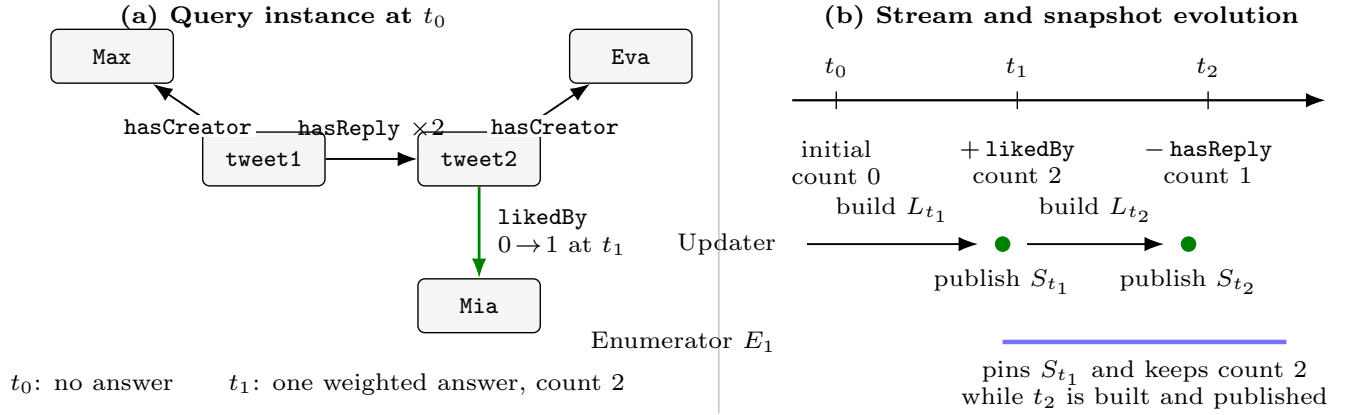

  \centering
  \resizebox{\textwidth}{!}{\TrieGSRunningOverview}
  \caption{Paper-wide running example in multigraph mode. The insertion at
  $t_1$ creates one weighted answer with multiplicity two. Enumerator $E_1$
  pins the published state $S_{t_1}$, so it continues to observe count two
  while update workers construct and publish the state for $t_2$.}
  \label{fig:running-overview}
\end{figure*}

The systems problem considered here begins when maintenance and enumeration run
on different shared-memory threads. In Figure~\ref{fig:running-overview}, the
update at $t_2$ may modify the reply branch, its parental projection, the root
view, and the answer index while $E_1$ is still traversing the result for
$t_1$. Independently snapshotting these relations can yield a fractured
state---for example, the root after $t_2$ together with a child and answer index
from $t_1$. Such a combination belongs to no graph snapshot.

\paragraph*{Why shared-memory scale-up?}
Microarchitectural studies indicate that general-purpose stream-processing
engines can leave substantial single-node performance headroom. In a controlled
queue experiment, Zeuch et al.~\cite{DBLP:journals/pvldb/ZeuchBRMKLRTM19} report approximately 3\,GB/s for the tested Java
queue and up to 7\,GB/s for the strongest tested lock-free C++ queues on a
machine with 28\,GB/s measured memory bandwidth; they identify queue-mediated
operator communication, synchronization, and data representation as potential
scale-up bottlenecks~\cite{analyzing-vldb-spe-multicores}. These measurements do
not imply that LFNT is faster, but they motivate testing whether a
query-specialized representation reduces synchronization, allocation, and data
movement. CROWN demonstrates complementary distributed parallelism: it reports
approximately linear speedup below 16 workers, followed by diminishing gains as
HyperCube scaling becomes sublinear and framework overhead dominates the
shortened useful work~\cite{WangHDY23}.

\paragraph*{Why enumerate a stable full result?}
A representative use case is periodic auditing or export of all currently
matching graph patterns. Such a scan may take longer than one update interval,
while pausing ingestion is unacceptable. A delta stream alone can serve this
workload only if the consumer maintains and versions a second complete answer
state. TrieGS instead lets the auditor pin one published query snapshot while
later updates continue; Section~\ref{sec:evaluation} measures when this service
is preferable to downstream delta consumption or aggregate-only SWAG summaries
rather than assuming it is universally superior.

\TrieGS addresses this problem by combining standard CDE representations with a
query-level publication protocol. Logically, a fixed query is compiled into a
free-connex decomposition and a witness-subtree enumerator of the kind used in
classical Yannakakis/CDE proofs. Dynamically, tuples carry multiplicity payloads
and updates maintain the corresponding factorized views using exact signed
deltas, as in factorized IVM. Physically, RDF terms and properties are first
encoded as stable integer identifiers, and each maintained view and access index
is implemented by a versioned \LFNT relation over those integer keys. After every
logical update epoch has propagated through all dependent views and indexes,
\TrieGS atomically publishes
\[
  \RootVec_t=\langle t,r^t_{R_1},\ldots,r^t_{R_m}\rangle,
\]
where the roots cover the whole compiled access representation. An enumerator
loads this vector once, pins it, and reads only versions reachable from it.
Later updates can then proceed without changing the enumerator's logical state.

\paragraph*{Proof architecture.}
We separate the proof into five layers so that each part matches an established
technique. First, a classical free-connex CDE argument constructs an abstract
preprocessed state with constant-delay cursor enumeration over a witness
subtree. Second, a multiset maintenance argument treats local tuples as
$\mathbb N$-payloads, uses signed deltas internally, and proves that each epoch
restores the view invariant. Third, an access-structure accounting layer states
which primary, secondary, and tertiary indexes are maintained and which
space--delay point the core theorem chooses. Fourth, dictionary-coded LFNT and
root-vector refinement prove that a published physical snapshot represents one
abstract CDE state. Fifth, morselized scheduling is proved only as a refinement
of the serial epoch semantics; it is not part of the query-class theorem.

\paragraph*{Theorem core and system scope.}
Our strongest no-output-materialization result is for fixed \emph{full}
acyclic conjunctive graph queries under multiset semantics, where every query
variable is distinguished. This class includes full-match paths, stars, trees,
combs, and neighborhood patterns. For projected free-connex queries, the
current construction explicitly materializes a count-aware answer relation;
its output-size-dependent update and space costs are part of the guarantee. We
do not claim a compact duplicate-free weighted iterator for arbitrary
projection. Filters, aggregate payloads, SWAG summaries, multi-query sharing,
compressed access structures, and cyclic bag modules are practical extensions
with separate or parameterized guarantees.

\paragraph*{Positioning.}
\TrieGS is neither a new one-shot parallel join nor a general-purpose lock-free
graph store. Parallel and worst-case-optimal join systems optimize a fresh
query execution~\cite{DBLP:conf/sigmod/LeisBK014,
DBLP:journals/pacmmod/WangWS23}; \TrieGS maintains persistent query-specific
state and overlaps view maintenance with snapshot enumeration. CROWN is the
closest query-processing predecessor: it gives join-free change propagation,
full-result and delta CDE, and a Flink DataStream implementation. \TrieGS
instead studies a shared-memory state in which an updater and a long-running
enumerator access versioned relations concurrently; the root-vector protocol
is the consistency boundary that CROWN does not need to define in its operator
execution model. We therefore evaluate a CQELS CROWN-style reimplementation and
refer to it as a faithful CROWN baseline only after it passes the validation
gate in Section~\ref{sec:evaluation}.

\paragraph*{Evaluation questions.}
The paper tests four core questions. First, does root-vector publication prevent
fractured query snapshots under skewed updates, deletions, windows, self-joins,
and long readers? Second, does \LFNT improve the relevant
latency--memory--freshness frontier over simpler integer-key and object-key
versioned maps? Third, is concurrent full-result enumeration a real workload
need, or can a delta-only or SWAG aggregate consumer provide equivalent
application behavior at lower cost? Fourth, does mixed update/enumeration
execution scale with physical cores while preserving the same tail-latency,
freshness, and memory targets? Section~\ref{sec:evaluation} turns these into
pre-registered experiments and rejection criteria rather than using raw
throughput as a proxy for all four.

\paragraph*{Contributions.}
This report makes the following contributions.
\begin{enumerate}[leftmargin=*]
  \item We propose time-indexed streaming-graph and multiset query semantics,
  including a precise enumeration contract, dictionary-encoding refinement, and
  a reduction of sliding windows to insertion and expiration updates.

  \item We present an abstract CDE layer based on free-connex decompositions,
  witness-subtree enumeration, and local fetch-first/fetch-next access
  structures, making the constant-delay part of the proof standard and
  implementation-independent.

  \item We propose a dynamic multiset maintenance layer with count-aware payloads,
  exact delta-view support enumeration, and $\OQ(A_t)$ affected-state update
  cost. Projected free-connex queries use an explicit $\mathsf{Ans}_Q$ relation.

  \item We make the access-structure cost explicit: core \TrieGS chooses the
  fully indexed constant-delay point, while compressed space--delay variants,
  integer-layout choices, and SWAG aggregate summaries are separated from the
  main theorem.

  \item We identify exactly what is inherited from C-trie/cache-trie work and
  add the missing query-level mechanism: immutable multiplicity payloads,
  epoch-ordered root-vector publication, and snapshot-safe reclamation.

  \item We define an evaluation methodology and a CQELS-based implementation
  plan that compare \LFNT with a CROWN-style baseline subject to an explicit
  validation gate, object-key and integer-key versioned indexes, SWAG aggregate
  baselines, and delta-only maintenance while reporting tail latency, staleness,
  retained memory, reclamation debt, and marginal cost per query in addition to
  throughput.
\end{enumerate}

\section{Semiring Stream Semantics}
\label{sec:preliminaries}

This section fixes the logical semantics that the later CDE, maintenance, and
snapshot-refinement layers implement. We use the counting semiring
$(\mathbb N,+,\cdot,0,1)$ for published query states and signed integers only
for transient maintenance deltas.

\subsection{Labeled Multigraphs and Logical Time}

\begin{definition}[Labeled multigraph]
Let $\mathcal{L}$ be a finite set of edge labels. A labeled multigraph is
$G=(N,E,\rho,\lambda_E)$, where $N$ is a finite node set, $E$ is a finite edge
set, $\rho:E\rightarrow N\times N$ maps an edge to its endpoints, and
$\lambda_E:E\rightarrow\mathcal{L}$ maps an edge to its label.
\end{definition}

For $\ell\in\mathcal{L}$ and $a,b\in N$, the multiplicity of the assertion
$\ell(a,b)$ is
\[
  \mult_G(\ell(a,b))
  =|\{e\in E\mid \rho(e)=(a,b),\ \lambda_E(e)=\ell\}|.
\]
Thus parallel edges are represented by nonnegative integer multiplicities.
Equivalently, each base relation is a finite-support map
\[
  R_G:\mathrm{Dom}(R)\rightarrow\mathbb N.
\]

Let $\mathbb{T}$ be a totally ordered set of logical times. A graph update
stream is
\[
  \mathcal{U}=\langle (t_1,U_{t_1}),(t_2,U_{t_2}),\ldots\rangle,
  \qquad t_1\prec t_2\prec\cdots,
\]
where $U_t$ is a finite batch of updates. Each update is
$(\sigma,\ell(a,b))$ with $\sigma\in\{+1,-1\}$. Let $G_t$ denote the graph
after all batches through $t$ have been applied. We assume valid deletions, or
equivalently reject deletions that would make an assertion count negative.

For any maintained relation $R$, its signed delta at time $t$ is
\[
  \delta_tR(\bar a)
  =\mult_{R,G_t}(\bar a)-\mult_{R,G_{\pred(t)}}(\bar a).
\]
Stored published multiplicities are in $\mathbb N$, while internal deltas are in
$\mathbb Z$.

\paragraph*{Event time and processing time.}
The logical time used by the proof orders published query states. It may be an
event-time watermark or a processing-time epoch. Event-time systems must not
publish $t$ until their lateness policy guarantees that the state for $t$ will
not be modified retroactively.

\subsection{Conjunctive Queries as Counting Maps}

A conjunctive graph query has the form
\[
  Q:q(\bar x)\leftarrow \alpha_1,\ldots,\alpha_k,
\]
where every $\alpha_i$ is a labeled edge atom and every distinguished variable
in $\bar x$ appears in the body. A valuation $\theta$ maps variables to graph
nodes and fixes constants.

At logical time $t$, the weight of a valuation is
\[
  \omega_t(\theta)=\prod_{i=1}^{k}\mult_{G_t}(\theta(\alpha_i)).
\]
For a head tuple $\bar a$, the answer multiplicity is
\begin{equation}
  \mu_t^Q(\bar a)
  =\sum_{\theta:\theta(\bar x)=\bar a}\omega_t(\theta).
  \label{eq:bag-query-semantics}
\end{equation}
Hence joins multiply multiplicities, while projection or marginalization sums
over witnesses. Bag semantics and multiset semantics are used interchangeably.

\paragraph*{Running example continued.}
At $t_1$ in Figure~\ref{fig:running-overview}, the displayed tuple has one node
valuation, but that valuation has weight
\[
  1\cdot 2\cdot 1\cdot 1=2.
\]
The factor two is the active multiplicity of
$\mathsf{hasReply}(\mathsf{tweet1},\mathsf{tweet2})$. When one reply event
expires at $t_2$, the same valuation has weight
$1\cdot 1\cdot 1\cdot 1=1$. This example is intentionally in multigraph mode:
it makes visible why a base delta of one may induce a larger derived delta.

\paragraph*{Enumeration contract.}
Internally, \TrieGS enumerates the support of the answer together with exact
counts: each pair $(\bar a,\mu_t^Q(\bar a))$ with positive multiplicity is
emitted once. Delay is measured between consecutive tuple--multiplicity pairs.
An application can expand a pair of multiplicity $d$ into $d$ rows when
required.

\paragraph*{RDF/SPARQL-compatible mode.}
SPARQL solution mappings use multiset semantics, but an RDF graph is a set of
triples. Let $c_t(\tau)$ be the number of active stream events representing
triple $\tau$. In RDF-compatible mode, query evaluation uses Boolean support
$\mathbf{1}[c_t(\tau)>0]$, while the window manager retains $c_t(\tau)$ for
correct expiration. Multigraph mode instead uses $c_t(\tau)$ as the base
weight. The core theorem corresponds to the positive basic-graph-pattern
fragment under simple entailment; full SPARQL operators require separate
semantics.

\subsection{Dictionary Encoding}
\label{sec:dictionary-encoding}

RDF systems usually evaluate over dictionary identifiers rather than over RDF
term objects. TrieGS follows this convention in the physical layer. Let
\[
  \iota:N\cup\mathcal L\rightarrow\mathbb N
\]
be an injective dictionary for all terms and labels reachable from a published
snapshot. A base assertion $\ell(a,b)$ is represented by the integer tuple
$(\iota(\ell),\iota(a),\iota(b))$; compiled views and indexes use fixed relation
or view identifiers together with integer values for the variables in their
chosen order.

\begin{lemma}[Dictionary refinement]
\label{lem:dictionary-refinement}
Replacing every term and label in the active graph, query constants, view keys,
and answer bindings by its dictionary identifier preserves equality, joins,
projections, multiplicities, and answer weights.
\end{lemma}

\begin{proof}
Conjunctive graph queries use equality of labels, constants, and shared
variables. Since $\iota$ is injective, two terms are equal iff their identifiers
are equal. Therefore the set of valuations, their joined atoms, and every
projection witness class are unchanged by encoding. Payload arithmetic is
performed on the same multiplicities, so Equation~\ref{eq:bag-query-semantics}
is preserved.
\end{proof}

The theorem needs only injectivity. Hardware locality requires more: the
assigned identifiers may or may not correlate with join access patterns. The
implementation therefore treats stable integer encoding as the correctness
layer and evaluates packed layouts, local dense remapping, and prefix ownership
as physical optimizations. Identifier reuse is unsafe while a published snapshot
can still reference the old term; the first implementation should use monotonic
or epoch-safe identifiers.

\subsection{Sliding Windows}

For window width $W$, let $G_t^W$ contain exactly the edge events whose event
timestamps lie in $(t-W,t]$. The window manager translates an event
$(\mathsf{id},\ell(a,b),\tau)$ into an insertion at $\tau$ and an expiration
deletion at $\tau+W$. Multiple identical events retain distinct identifiers so
that one expiration decrements the assertion count rather than removing all
copies. The transition at $t_2$ in Figure~\ref{fig:running-overview} is exactly
such an expiration. The core theorem applies to $G_t^W$ because window movement
is reduced to ordinary multiset updates.

\subsection{Acyclicity, Free-Connexity, and Cost Convention}

The query hypergraph $H_Q$ has one vertex per variable and one hyperedge
$\Vars(\alpha_i)$ per atom. A conjunctive query is acyclic if $H_Q$ has a join
tree. It is \emph{free-connex} if adding one hyperedge containing all
distinguished variables preserves acyclicity~\cite{bagan2007acyclic}. Every full
acyclic query is free-connex because the added head hyperedge contains all body
hyperedges. Projected free-connex queries are treated separately because they
require witness grouping.

The query $Q$ and its compiled plan are fixed. We write $\OQ(\cdot)$ when the
hidden constant may depend on the atoms, variables, views, projections,
indexes, access-structure arity, and join-tree height of $Q$, but never on graph
size, stream length, materialized-state size, or output cardinality. Let
$T_{\mathsf{dict}}$ denote the cost of an exact relation-index dictionary
operation, and $T_{\mathsf{fetch}}$ the cost of one local extension-iterator
step. The cache-trie assumptions give $T_{\mathsf{dict}}=\OQ(1)$ for basic
exact-key operations; the fully indexed TrieGS access contract supplies
$T_{\mathsf{fetch}}=\OQ(1)$. Without these assumptions, the update and delay
bounds inherit the selected backend's $T_{\mathsf{dict}}$ and $T_{\mathsf{fetch}}$.

\subsection{Published Snapshot Semantics}

Let $L_t$ be the abstract logical view state after all effects of epoch $t$ have
completed, let $C_t$ be the compiled access representation of that view state,
and let $S_t$ be a physical versioned snapshot that realizes $C_t$. An
enumeration job $e$ is snapshot-consistent if there exists a published time
$\tau(e)$ such that every relation and index read by $e$ belongs to
$S_{\tau(e)}$ and therefore represents $C_{\tau(e)}$ and $L_{\tau(e)}$. The job
then enumerates $Q(G_{\tau(e)})$. It need not observe the newest wall-clock
update, but it must never mix epochs.

\section{Abstract CDE State and View Invariant}
\label{sec:logical-join-trie}

This section defines the abstract state that \TrieGS maintains. The first part
follows the classical free-connex CDE proof: a decomposition gives a connected
witness subtree, local interface iterators support fetch-first/fetch-next, and
a cursor enumerator outputs answers with constant delay. The second part adds
counting payloads and view equations that dynamic maintenance will preserve.
Concurrency and LFNT appear only after this abstract state is fixed.

\subsection{Witness-Subtree CDE State}

Let $Q$ be a fixed acyclic conjunctive query, and let $T_Q$ be a rooted join
or generalized hypertree decomposition of width one. Each node $u$ has a bag
$\chi(u)$ and is assigned one or more guard atoms whose variables are contained
in $\chi(u)$. For every variable, the nodes whose bags contain that variable
form a connected subtree.

For a projected free-connex query, the decomposition is chosen after adding the
head hyperedge. A connected witness subtree $U_Q\subseteq T_Q$ covers all
distinguished variables. For a full acyclic query, every variable is
distinguished and the whole rooted tree can be used as the witness subtree. The
abstract CDE state contains, for each witness node $u$, a semijoin-reduced local
relation $C_u$ and an access structure by the parent interface
\[
  I_u=\chi(u)\cap\chi(\mathsf{parent}(u)).
\]
For the root, $I_u=\emptyset$.

\begin{definition}[Local access contract]
For each witness node $u$, the access structure supports
\begin{align*}
  \mathsf{fetchFirst}_u(\bar b) &\quad\text{for the first tuple in } C_u
     \text{ matching interface binding }\bar b,\\
  \mathsf{fetchNext}_u(\bar b,\bar c) &\quad\text{for the next such tuple after }\bar c.
\end{align*}
Both operations skip dead tuples, return each compatible tuple once, and cost
$\OQ(1)$ in the fully indexed core construction.
\end{definition}

The abstract enumerator stores one cursor per witness node. It initializes the
root cursor with $\mathsf{fetchFirst}$, recursively initializes child cursors
from their parent-interface bindings, emits a tuple when all witness cursors are
valid, and advances the last cursor that still has a next value, resetting its
descendants. This is the standard constant-delay cursor argument for
free-connex acyclic CQs.

\begin{theorem}[Abstract CDE enumeration]
\label{thm:abstract-cde-enumeration}
For a fixed full acyclic query, a semijoin-reduced witness-subtree state with
local access structures satisfying the contract above enumerates the support of
$Q(G)$ with $\OQ(1)$ delay after preprocessing. For a projected free-connex
query, the same statement holds when the access state is a duplicate-free
support representation of the projected head tuples.
\end{theorem}

\begin{proof}[Proof sketch]
Semijoin reduction removes every local tuple that cannot be extended to a full
homomorphism. The witness subtree covers the distinguished variables and the
running-intersection property ensures that a top-down cursor assignment is
consistent on shared variables. Every output corresponds to exactly one tuple
choice at each witness node, and the cursor discipline visits each such choice
once. The number of witness nodes and local fetch operations between two
outputs depends only on the fixed query, giving $\OQ(1)$ delay.
\end{proof}

\subsection{Counting Payloads and View Recurrence}

TrieGS maintains a counting version of the abstract CDE state. For every node
$u$, $V_u(\chi(u))$ stores the multiset payload of the subquery rooted at $u$.
If $c$ is a child of $u$, its interface with the parent is
\[
  I_{c,u}=\chi(c)\cap\chi(u).
\]
The parental projection stores a witness count, not only a Boolean semijoin
flag:
\begin{equation}
  \mu_t^{P_{c\rightarrow u}}(\bar b)
  =\sum_{\bar a:\bar a|_{I_{c,u}}=\bar b}
    \mu_t^{V_c}(\bar a).
  \label{eq:parent-projection}
\end{equation}
The viability predicate is derived as
$\mu_t^{P_{c\rightarrow u}}(\bar b)>0$. Retaining the count is necessary for
bag semantics and deletions: a change from three witnesses to two does not
invalidate a branch, whereas a change from one to zero does.

A general normalized view rule has the form
\[
  V(\bar x)\leftarrow B_1(\bar y_1),\ldots,B_k(\bar y_k).
\]
Its exact logical multiplicity is
\begin{equation}
  \mu_t^V(\bar a)
  =\sum_{\theta:\theta(\bar x)=\bar a}
    \prod_{j=1}^{k}\mu_t^{B_j}(\theta(\bar y_j)).
  \label{eq:view-sum-product}
\end{equation}
This is the counting-semiring payload version of the abstract local CDE state.

\paragraph*{Guard-normal compiled rules.}
The compiler represents every maintained acyclic rule in the equivalent form
\begin{equation}
 r:\quad V(\bar X)\leftarrow
 G(\bar X),P_1(\bar I_1),\ldots,P_m(\bar I_m),
 \qquad \bar I_i\subseteq\bar X,
 \label{eq:guard-normal-rule}
\end{equation}
where $G$ is the \emph{guard} containing all head variables and every $P_i$ is
a count-aware projection whose eliminated variables have already been summed.
If several local atoms are needed to cover $\bar X$, their local product is
materialized as $G$. Equation~\ref{eq:view-sum-product} then reduces for a
grounded head tuple $\bar a$ to
\begin{equation}
  \mu_t^V(\bar a)
  =\mu_t^G(\bar a)
   \prod_{i=1}^{m}\mu_t^{P_i}(\bar a|_{\bar I_i}).
  \label{eq:guard-normal-product}
\end{equation}
This form makes the dynamic affected-binding construction in
Section~\ref{sec:maintenance-enumeration} explicit rather than treating it as
an oracle.

For each interface $\bar I_i$, the plan contains an update/access path
\begin{equation}
  \idx{G,\bar I_i}(\bar b)
  =\{\bar a\mid \mu_t^G(\bar a)>0,
                    \bar a|_{\bar I_i}=\bar b\}.
  \label{eq:guard-interface-index}
\end{equation}
In the access-structure terminology used below, exact payload lookup is a
primary index, top-down enumeration uses secondary indexes, and
Equation~\ref{eq:guard-interface-index} is also the tertiary/update index that
finds parent bindings affected by a changed projection.

\subsection{Generalized Multiset Assertions}

\begin{definition}[Generalized multiset assertion]
For a base relation, view, projection, answer relation, or auxiliary index
$R(\bar x)$, a generalized multiset assertion (Gma) is
$R(\bar a)=d$, where $d\in\mathbb{N}$ is the tuple multiplicity.
\end{definition}

Missing Gmas have multiplicity zero. The logical operations
$\getM(R(\bar a))$ and $\setM(R(\bar a),d)$ read and write this value.
Zero-valued physical tombstones may remain until reclamation, but logical
iterators skip them.

\subsection{Running Example}

The running query is
\begin{lstlisting}[caption={Running full acyclic query.},label={lst:running-query}]
SELECT ?x ?y ?v ?z ?w
WHERE {
  ?v :hasCreator ?x .
  ?z :hasCreator ?y .
  ?v :hasReply   ?z .
  ?z :likedBy    ?w .
}
\end{lstlisting}
A normalized program is
\begin{align*}
H^p(v)   &\leftarrow \mathsf{hasCreator}_v(v,x),\\
L^p(z)   &\leftarrow \mathsf{likedBy}(z,w),\\
V_1(v,z) &\leftarrow \mathsf{hasReply}(v,z),H^p(v),\\
V_1^p(z) &\leftarrow V_1(v,z),\\
V_2(z,y) &\leftarrow \mathsf{hasCreator}_z(z,y),L^p(z),\\
V_2^p(z) &\leftarrow V_2(z,y),\\
V_0(z)   &\leftarrow V_1^p(z),V_2^p(z).
\end{align*}
The two occurrences $\mathsf{hasCreator}_v$ and
$\mathsf{hasCreator}_z$ are logical aliases of the same base label. Thus the
running example already contains a self-join. A base update is fanned out to
both aliases with one epoch identifier, and publication waits until both
copies, their projections, and all dependent views have processed that epoch
(Section~\ref{sec:implementation}).

Figure~\ref{fig:running-view-evolution} grounds this program on the values from
Figure~\ref{fig:running-overview}. Before $t_1$, the right branch is not
viable because $L^p(\mathsf{tweet2})=0$. After the insertion,
\[
\begin{aligned}
V_1^p(\mathsf{tweet2}) &= 2,\\
V_2^p(\mathsf{tweet2}) &= 1,\\
V_0(\mathsf{tweet2}) &= 2\cdot 1=2.
\end{aligned}
\]
The same figure also previews the dynamic-maintenance rule: the input delta is
$+1$, but the exact root delta is $+2$ because the unchanged reply branch
contributes two witnesses.

\begin{figure*}[t]
  \centering
  \resizebox{\textwidth}{!}{\TrieGSViewEvolution}
  \caption{Counting view state and exact delta propagation for the running
  example. Dashed nodes are parental projections. Values in green change at
  $t_1$. The affected-binding indexes enumerate the support of the corresponding
  delta view; each affected Gma is recomputed and propagated as its own
  difference from the old value.}
  \label{fig:running-view-evolution}
\end{figure*}

The root view identifies viable values of $z$. A factorized full-query
enumerator descends through $V_1$, $V_2$, and the base relations to recover
$(x,y,v,z,w)$. In this concrete instance there is one distinct full assignment,
so its weight equals the displayed root count two. In general, the root count
must not be interpreted as one flat answer multiplicity unless the compiled
answer representation guarantees that correspondence.

\paragraph*{Projected queries.}
A projected query such as $Q'(x,y,w)$ requires a separate free-connex check. For
projected free-connex queries, the current implementation maintains an explicit
count-aware answer relation $\mathsf{Ans}_Q$; it does not claim that a compact
factorized representation automatically provides duplicate-free weighted
enumeration. If projection destroys free-connexity, the query is outside the
core result.

\subsection{Answer Relation and View Invariant}

Let $\mathcal{R}_Q$ be the finite set of relations and indexes maintained for
$Q$. For a projected free-connex query, the formal construction includes a
count-aware answer relation
\[
  \mathsf{Ans}_Q(\bar x)
\]
whose tuple multiplicity is exactly $\mu_t^Q(\bar a)$ from
Equation~\ref{eq:bag-query-semantics}. It is maintained as the final
projection/count view of the compiled program. Scanning its support emits each
distinct head tuple once with its exact multiplicity. Its output-size-dependent
update and space costs are included in the affected-state and materialized-state
bounds.

For full queries, where every query variable is distinguished, the
implementation instead uses a factorized generator: no existential witnesses
must be grouped into the same head tuple, and one trie path represents one node
valuation with its accumulated multiplicity.

\begin{definition}[View invariant]
At published time $t$, the view invariant holds if, for every
$R\in\mathcal{R}_Q$ and tuple $\bar a$,
\[
  \getM_t(R(\bar a))=\mu_t^R(\bar a),
\]
where the right-hand side is the sum-product value defined by the logical
query, projection, view, answer, or index semantics over $G_t$.
\end{definition}

The rest of the proof has one purpose: show that update propagation moves a
correct state $L_{\pred(t)}$ to a correct state $L_t$, that the compiled access
representation $C_t$ supports the required local fetch operations, and that
concurrent enumerators observe exactly one physical realization $S_t$ of that
state.

\subsection{Existence of a Guard-Normal Plan}

\begin{lemma}[Guard normalization for full acyclic queries]
\label{lem:guard-normalization}
Every full acyclic conjunctive query admits a rooted materialized-view program
whose maintained rules have the guard-normal form of
Equation~\ref{eq:guard-normal-rule}. Repeated relation atoms use distinct
logical aliases.
\end{lemma}

\begin{proof}
Acyclicity gives a join tree whose nodes are the body atoms. Root the tree at an
arbitrary atom. For a node with atom $R_v(\bar X_v)$, define the subquery view
$V_v(\bar X_v)$ using $R_v$ as its guard. For every child $c$, project the
child view onto the separator
$\bar I_{c,v}=\bar X_c\cap\bar X_v$ and use that count-aware projection as a
factor of $V_v$. The running-intersection property ensures that the child
subquery shares variables with the remaining query only through this
separator. Thus all eliminated child variables are summed in the projection,
each projection interface is contained in the guard variables, and the
grounded node value is the product in
Equation~\ref{eq:guard-normal-product}. Applying this construction bottom-up
gives guard-normal rules for every node. Since the query is full, traversing
the rooted views recovers every query variable without a projected head
aggregation. Aliasing repeated atoms preserves the join tree and the argument.
\end{proof}

For projected free-connex queries, the body-node construction still applies,
but the current implementation materializes the final count-aware answer
projection as stated in Corollary~\ref{cor:projected-materialized}.

\section{Dynamic Multiset Maintenance and Enumeration}
\label{sec:maintenance-enumeration}

The previous section defined the abstract CDE state. This section explains how
TrieGS maintains its counting payloads under insertions, deletions, batches, and
window expirations. Conceptually, a base update induces a delta view tree: a
changed leaf is joined with unchanged sibling views and projected toward the
root. TrieGS does not need to materialize every delta view. Instead, tertiary
access indexes enumerate exactly the parent bindings in the delta-view support;
for each such binding, \textsc{SetExact} recomputes the new payload and
propagates the signed difference.

\subsection{Serializable Update Epochs}

For the core proof, epochs are published in logical-time order. Work inside one
epoch may execute in parallel, but its effect must be equivalent to some
deterministic sequential order. This requirement is particularly important for
self-joins and batches, where one base relation may occur several times in the
same query.

\begin{assumption}[Serializable epochs]
The concurrent execution of update epoch $t$ is equivalent to applying all
updates in $U_t$ and all induced Gma changes in a deterministic sequential
order. Epoch $t$ is published only after every induced change has completed.
\label{ass:serializable-epochs}
\end{assumption}

\subsection{Delta-View Support and Affected Bindings}

A changed base or derived Gma may affect projections and parent views. The
compiler uses the guard-normal rule and interface indexes from
Equations~\ref{eq:guard-normal-rule}--\ref{eq:guard-interface-index} to derive
the affected head bindings directly.

\paragraph*{Exact derived updates.}
For a Gma $R(\bar a)$ whose exact new value is $d_{\newv}$, the logical
operation \textsc{SetExact} reads $d_{\oldv}=\getM(R(\bar a))$, computes
$\delta=d_{\newv}-d_{\oldv}$, and does nothing when $\delta=0$. Otherwise it
installs $d_{\newv}$ in the working version and enqueues $(R(\bar a),\delta)$
for every dependent projection and view. This operation is the only way a
derived Gma propagates upward.

\paragraph*{Affected-binding construction.}
Consider a guard-normal rule
\[
 r:\quad V(\bar X)\leftarrow
 G(\bar X),P_1(\bar I_1),\ldots,P_m(\bar I_m).
\]
The construction is evaluated in the current working version after the changed
factor has received its new value.

If the guard changes at binding $\bar c$, then
\[
 \mathsf{Affected}_r(G(\bar c))=
 \begin{cases}
 \{\bar c\}, & \text{if }\prod_i\mu(P_i(\bar c|_{\bar I_i}))>0,\\
 \emptyset, & \text{otherwise.}
 \end{cases}
\]
If projection $P_j$ changes at interface binding $\bar c$, then
\begin{align}
 \mathsf{Affected}_r(P_j(\bar c))
 =\{\bar a\in \idx{G,\bar I_j}(\bar c)\mid
       \mu(P_i(\bar a|_{\bar I_i}))>0
       \text{ for all }i\neq j\}.
 \label{eq:affected-bindings}
\end{align}
Thus a projection change probes the guard by its interface and filters only by
unchanged sibling viability. Logical aliases distinguish repeated atoms in a
self-join, so the construction is applied once per affected alias.

\begin{lemma}[Delta-view support]
\label{lem:affected-bindings}
Let one factor of a grounded guard-normal rule change by a nonzero signed delta
while the remaining factors are fixed. The construction above returns exactly
the head bindings whose value under Equation~\ref{eq:guard-normal-product}
changes; equivalently, it enumerates the support of the corresponding local
delta view.
\end{lemma}

\begin{proof}
First suppose the guard changes at $\bar c$. A grounded head binding uses the
changed guard factor iff it is exactly $\bar c$. Its product changes iff every
projection factor is positive; otherwise the product is zero both before and
after the guard update. The guard case therefore returns precisely the changed
head tuple, or none.

Now suppose $P_j(\bar c)$ changes. A grounded head binding $\bar a$ uses this
factor iff $\bar a|_{\bar I_j}=\bar c$. By
Equation~\ref{eq:guard-interface-index}, the tertiary index returns exactly
such bindings having a positive guard factor. For any returned binding, the
remaining filter requires every unchanged sibling factor to be positive. The
product of the unchanged factors is therefore positive, so multiplying it by a
nonzero change in $P_j(\bar c)$ changes the value of $V(\bar a)$; this proves
soundness. Conversely, if $V(\bar a)$ changes, it must use the changed factor,
so $\bar a|_{\bar I_j}=\bar c$, and all unchanged factors must be positive;
otherwise its product remains zero. The guard index and sibling filter therefore
return $\bar a$, proving completeness.
\end{proof}

\subsection{Epoch Processing}

\begin{algorithm}[H]
\caption{Process one logical update epoch}
\label{alg:process-epoch}
\begin{algorithmic}[1]
\Procedure{ProcessEpoch}{$U_t,t$}
  \State create a working version from the latest published state
  \ForAll{base updates $(\sigma,A(\bar a))\in U_t$ in epoch order}
    \State \Call{SetExact}{\ensuremath{A(\bar a),\getM(A(\bar a))+\sigma,t}}
  \EndFor
  \While{the epoch work queue is nonempty}
    \State pop a changed Gma $A(\bar a)$ with signed delta $\delta_A$
    \ForAll{projection rules $P(\bar y)\leftarrow A(\bar x)$}
      \State $\theta\gets\mathsf{CreateMap}(\bar x,\bar a)$
      \State $b\gets\theta(\bar y)$
      \State \Call{SetExact}{\ensuremath{P(b),\getM(P(b))+\delta_A,t}}
    \EndFor
    \ForAll{parent rules $r$ containing $A$ in their body}
      \ForAll{$\theta\in\mathsf{Affected}_r(A(\bar a))$}
        \State $d\gets\mathsf{EvaluateRule}_r(\theta)$
        \State \Call{SetExact}{\ensuremath{\mathsf{head}_r(\theta),d,t}}
      \EndFor
    \EndFor
  \EndWhile
  \State mark epoch $t$ complete
\EndProcedure
\end{algorithmic}
\end{algorithm}

The key difference from forwarding an input edge delta is the parent-rule loop:
a derived view propagates its own exact difference $d_{\newv}-d_{\oldv}$. Under
bag semantics, an insertion of one base edge can change a parent count by more
than one because it joins with several witnesses.

\paragraph*{Running example continued.}
Panel~(b) of Figure~\ref{fig:running-view-evolution} is the concrete execution
of Algorithm~\ref{alg:process-epoch} at $t_1$. The insertion first changes
$L^p(\mathsf{tweet2})$ and the right branch by $+1$. Equation~\ref{eq:affected-bindings}
uses the $z$-index of the guard to find the one root binding
$z=\mathsf{tweet2}$. Re-evaluating the root multiplies that new right-branch
count by the two reply witnesses on the left, so \textsc{SetExact} installs
$V_0(\mathsf{tweet2})=2$ and propagates a root delta of $+2$. At $t_2$,
expiration of one reply event executes the symmetric trace:
\[
\begin{aligned}
\mathsf{hasReply}(\mathsf{tweet1},\mathsf{tweet2}) &: 2\rightarrow 1,\\
V_1(\mathsf{tweet1},\mathsf{tweet2}) &: 2\rightarrow 1,\\
V_1^p(\mathsf{tweet2}) &: 2\rightarrow 1,\\
V_0(\mathsf{tweet2}) &: 2\rightarrow 1.
\end{aligned}
\]
Thus window expiration needs no special query rule; it is an exact negative
base update followed by the same propagation mechanism.

\begin{lemma}[Update correctness]
\label{lem:update-correctness}
If the view invariant holds at $\pred(t)$, then after
Algorithm~\ref{alg:process-epoch} completes, it holds at $t$.
\end{lemma}

\begin{proof}
Base Gmas are updated exactly according to $U_t$. Equation~\ref{eq:parent-projection}
is linear, so a child delta contributes the same signed delta to its grounded
parental projection. For every derived rule,
Lemma~\ref{lem:affected-bindings} proves that the parent-rule loop visits every
and only grounded head binding whose product can change. Each visited binding
is re-evaluated from the current working version using
Equation~\ref{eq:guard-normal-product}, and \textsc{SetExact} propagates its
exact difference from the old value. Induction from leaves to the root
establishes the invariant for every maintained relation. Serializable epoch
execution ensures that parallel scheduling has the same final state as the
logical sequence used by the proof.
\end{proof}

\begin{theorem}[Dynamic multiset maintenance]
\label{thm:dynamic-multiset-maintenance}
For fixed $Q$, every completed epoch $t$ restores the counting view invariant
and costs $\OQ(A_t)$ relation operations, where $A_t$ counts the affected Gmas,
projection updates, and compatible extensions touched by
Algorithm~\ref{alg:process-epoch}.
\end{theorem}

\begin{proof}
Correctness is Lemma~\ref{lem:update-correctness}. The number of relation
operations is proportional, up to constants depending only on the compiled
query, to the Gmas and indexed compatible extensions enumerated by the affected
work. This quantity is $A_t$ by definition.
\end{proof}

The result is not a constant-update theorem. A single edge may have many
compatible witnesses, and $A_t$ can be large.

\subsection{Enumeration over a Maintained State}

Algorithm~\ref{alg:enumerate} reads one fixed query snapshot $S_t$. For a full
query, a factorized generator implements the abstract witness-subtree cursor
from Theorem~\ref{thm:abstract-cde-enumeration}. For a projected free-connex
query, the current construction instead scans the explicitly materialized
$\mathsf{Ans}_Q(\bar x)$ relation, whose support contains each distinct head
tuple once and whose payload contains the exact multiset count. Zero-count Gmas
are skipped.

\begin{assumption}[Indexed extension access]
\label{ass:indexed-extension-access}
Every primary, secondary, and tertiary access path selected by the compiler
implements its advertised lookup or compatible-extension iterator without
duplicates and with $\OQ(1)$ local delay in the fully indexed core
configuration.
\end{assumption}

\begin{algorithm}[H]
\caption{Enumerate one published query snapshot}
\label{alg:enumerate}
\begin{algorithmic}[1]
\Procedure{Enumerate}{$Q$}
  \State $\mathsf{rv}\gets\Call{AcquirePublishedRootVector}{}$
  \State \Call{Pin}{\ensuremath{\mathsf{rv}}}
  \State $S\gets\Call{SnapshotFromRootVector}{\mathsf{rv}}$
  \ForAll{$(\bar a,d)$ produced by the answer iterator over $S$}
    \State \textbf{yield} $(\bar a,d)$
  \EndFor
  \State \Call{Unpin}{\ensuremath{\mathsf{rv}}}
\EndProcedure
\end{algorithmic}
\end{algorithm}

For the running example at $t_1$, Algorithm~\ref{alg:enumerate} emits
\[
\bigl((\mathsf{Max},\mathsf{Eva},\mathsf{tweet1},
\mathsf{tweet2},\mathsf{Mia}),2\bigr)
\]
once. It does not materialize two identical output rows internally. An
application requiring expanded bag output can repeat the tuple twice at its
interface boundary.

\begin{lemma}[Enumeration correctness over a logical snapshot]
\label{lem:logical-enumeration-correctness}
If snapshot $S_t$ represents a state satisfying the view invariant, then the
full-query generator emits exactly the nonzero pairs
$(\bar a,\mu_t^Q(\bar a))$. Scanning an explicit projected answer relation has
the same property for the projected head.
\end{lemma}

\begin{proof}
For a full query, every variable is distinguished. The witness-subtree cursor
therefore assigns every query variable, and the connectedness condition makes
shared-variable assignments consistent. Local access structures enumerate every
compatible extension exactly once. Products of Gma counts give the valuation
weight from Equation~\ref{eq:bag-query-semantics}; because there are no
existential variables to group, the generator emits the valuation once with
that weight. For a projected query, the view invariant states directly that
every nonzero $\mathsf{Ans}_Q(\bar a)$ payload equals $\mu_t^Q(\bar a)$, so
scanning its support gives the required weighted pairs.
\end{proof}

\begin{lemma}[Logical constant delay]
\label{lem:logical-constant-delay}
Under Assumption~\ref{ass:indexed-extension-access}, a fixed full acyclic query
has $\OQ(1)$ delay per emitted pair. An explicit projected answer relation also
has $\OQ(1)$ scan delay.
\end{lemma}

\begin{proof}
The number of witness-subtree cursors, child interfaces, and index probes is
bounded by the fixed compiled query. Semijoin/count viability prevents traversal
into dead subtrees. Therefore the work before the first output and between
consecutive outputs is a fixed number of local iterator steps and probes. For
an explicit answer relation, the iterator advances directly between nonzero
support entries.
\end{proof}

\section{Access Structures, LFNT, and Snapshot Publication}
\label{sec:lfnt}

The previous sections defined an abstract CDE state $L_t$ and a dynamic
maintenance algorithm that restores it after each epoch. This section separates
two further layers. First, TrieGS compiles $L_t$ into an access representation
$C_t$ with primary, secondary, and tertiary indexes over dictionary identifiers.
Second, LFNT and the root-vector protocol provide a physical snapshot $S_t$ that
realizes $C_t$ for concurrent readers.

\subsection{Access-Structure Accounting}

For a fixed query and fixed access pattern, an access representation stores the
relations and indexes needed by the local CDE fetch operations. We distinguish:
\begin{description}[leftmargin=*]
  \item[Primary indexes] lookup the exact payload of a view or projection tuple;
  \item[Secondary indexes] enumerate compatible extensions during top-down
  snapshot enumeration;
  \item[Tertiary indexes] enumerate parent bindings affected by a changed guard
  or projection during bottom-up maintenance.
\end{description}

The core \TrieGS theorem chooses the fully indexed constant-delay point of this
design space. This is not the only possible point: compressed representations
of conjunctive-query results can trade space for larger access delay~\cite{deep2018compressed}.
We keep such compressed variants outside the theorem and account for the fully
indexed state explicitly.

Let $M_t$ be the support size of the maintained logical views and projections,
$I_t$ the support and overhead of the access indexes, and $R_t$ the old physical
versions retained by active readers. The snapshot-state space is reported as
\[
  S_t = \OQ(M_t+I_t+R_t),
\]
where the constant hides fixed query arities and index multiplicity. This is a
state-accounting identity rather than a claim that $M_t$ or $I_t$ is linear in
the input for all workloads.

\begin{proposition}[Access-structure contract]
\label{prop:access-accounting}
Let $C_t$ be a lossless access representation of the logical state $L_t$ for a
fixed query $Q$. If every local primary, secondary, and tertiary access selected
by the compiled plan has delay at most $\delta_Q$, then enumeration over $C_t$
emits the same weighted answer pairs as $L_t$ with delay $\OQ(\delta_Q)$. In the
fully indexed core configuration, $\delta_Q=\OQ(1)$.
\end{proposition}

\begin{proof}
Losslessness ensures that every abstract local relation and payload needed by
Theorem~\ref{thm:abstract-cde-enumeration} and Lemma~\ref{lem:logical-enumeration-correctness}
is represented in $C_t$. The abstract enumerator performs only a fixed number of
local accesses between outputs. Replacing each abstract access by an operation
of delay $\delta_Q$ gives delay $\OQ(\delta_Q)$ and preserves the emitted
payloads.
\end{proof}

\subsection{Dictionary-Coded LFNT Keys}

By Lemma~\ref{lem:dictionary-refinement}, RDF terms and labels may be replaced
by stable dictionary identifiers without changing query semantics. Therefore the
hot physical representation is integer-centric. For a relation $R(a_1,\ldots,a_k)$
and fixed attribute order $\pi$, its wordified key is
\[
  \mathsf{Wrd}_{R,\pi}(R(a_1,\ldots,a_k))
  =(\mathsf{id}(R),\iota(a_{\pi(1)}),\ldots,\iota(a_{\pi(k)})).
\]
The map is injective, so a Gma operation is a dictionary operation on the
underlying trie. Different access orders are represented by a fixed set of
primary, secondary, and tertiary LFNT indexes chosen at query compilation. A
high-performance implementation should store these keys as primitive fixed-arity
integer tuples, not as boxed RDF terms or comparator-heavy object arrays.

\paragraph*{Running example continued.}
For the view assertion
$V_1(\mathsf{tweet1},\mathsf{tweet2})=2$, the natural attribute order produces
\[
  \mathsf{Wrd}(V_1(\mathsf{tweet1},\mathsf{tweet2}))
  =(\mathsf{id}(V_1),\iota(\mathsf{tweet1}),\iota(\mathsf{tweet2})).
\]
A secondary index may use the order
$(\mathsf{id}(V_1),\iota(\mathsf{tweet2}),\iota(\mathsf{tweet1}))$ so that
enumeration starting from the root binding $z=\mathsf{tweet2}$ can find
compatible $v$ values directly. Figure~\ref{fig:lfnt-physical} separates this
logical wordification from the copy-on-write versioning mechanism: only the
changed path is replaced, while unaffected subtries and the older payload remain
reachable from a pinned snapshot.

\begin{figure*}[t]
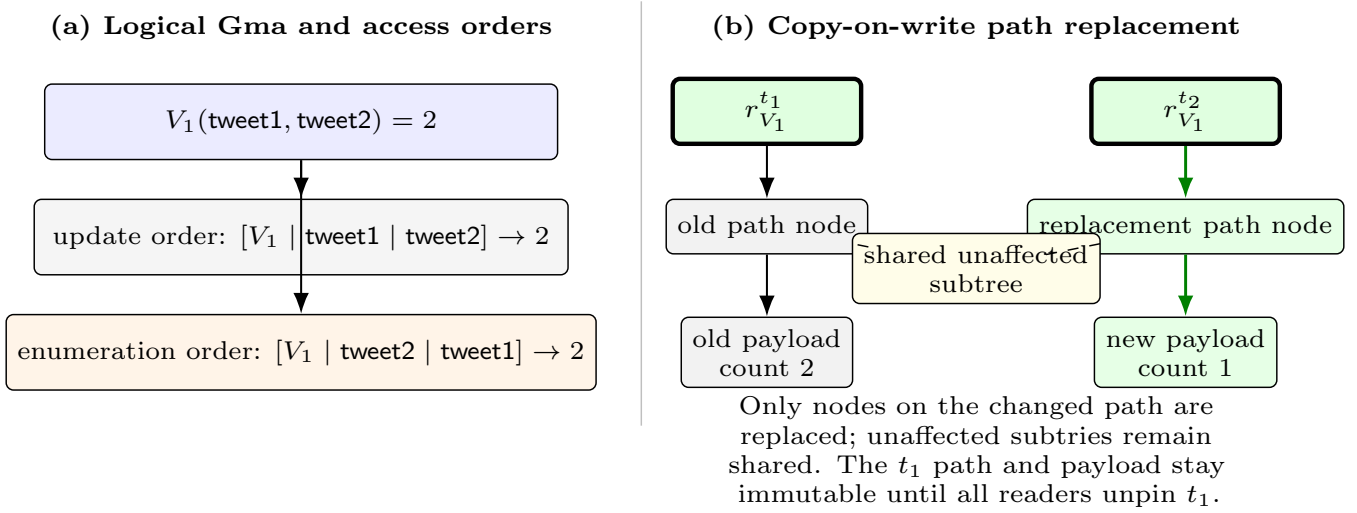

  \centering
  \resizebox{\textwidth}{!}{\TrieGSLFNTPhysical}
  \caption{Logical-to-physical LFNT mapping for the running example. A fixed
  query requires only a fixed number of word components, but this does not
  imply a one-memory-access or one-CAS lookup. In the copy-on-write realization,
  nodes on the changed path are replaced while unaffected subtries remain
  shared; the old path and multiplicity payload remain visible to pinned
  snapshots.}
  \label{fig:lfnt-physical}
\end{figure*}

\subsection{Relation-Level Guarantees from C-trie}

TrieGS uses the C-trie/cache-trie results only at the level of one maintained
relation or access index. The C-trie proof supplies linearizable and lock-free
exact-key dictionary operations and non-blocking snapshots; the cache-trie
analysis supplies expected or amortized constant-time basic dictionary
operations under its hashing assumptions~\cite{DBLP:conf/ppopp/ProkopecBBO12,
DBLP:conf/ppopp/Prokopec18}. The removing and compaction result supplies the
relation-level mechanism for deleting zero-payload keys without leaving
unbounded dead structure~\cite{DBLP:conf/europar/Prokopec18}. These 
results do not have direct connections to free-connex enumeration, affected-binding completeness,
secondary/tertiary extension delay, or query-level snapshot consistency.

\begin{lemma}[Relation-level LFNT refinement]
\label{lem:gma-instantiation}
Assume each LFNT index satisfies the underlying C-trie/cache-trie contracts and
wordification is injective. Then exact-key lookup, insertion, deletion, and
payload replacement on one abstract index refine linearizable LFNT dictionary
operations with cost $T_{\mathsf{dict}}$. Under the cache-trie assumptions,
$T_{\mathsf{dict}}=\OQ(1)$. A relation snapshot exposes one immutable logical
version of that index.
\end{lemma}

\begin{proof}
Lemma~\ref{lem:dictionary-refinement} and wordification preserve the equality
classes of RDF terms and view tuples. Therefore each abstract Gma maps to one
physical integer key. The underlying C-trie/cache-trie theorems apply to exact
operations on that key. Snapshot immutability is relation-local: it freezes one
index root, not an entire query state.
\end{proof}

The duplicate-free delay guarantee for compatible-extension scans is a separate
TrieGS access-path obligation stated in Assumption~\ref{ass:indexed-extension-access}.
In the notation of Section~\ref{sec:preliminaries}, cache-trie results justify
$T_{\mathsf{dict}}=\OQ(1)$ for exact-key operations, while TrieGS's compiled
secondary and tertiary indexes must also provide $T_{\mathsf{fetch}}=\OQ(1)$ in
the fully indexed configuration. Hence update work costs
\[
  \OQ\bigl(A_t\cdot(T_{\mathsf{dict}}+T_{\mathsf{fetch}})\bigr),
\]
and enumeration delay is $\OQ(T_{\mathsf{fetch}})$; the main theorem
instantiates both local costs as $\OQ(1)$.

\subsection{Memory-Hierarchy Trade-off}
\label{sec:memory-hierarchy}

The word ``cache'' in Cache-Trie denotes an auxiliary pointer array that
shortcuts logical trie traversal; its expected operation bound is not a claim
that every access hits a hardware cache. A pointer-based versioned LFNT can
incur dependent pointer dereferences, capacity and TLB misses, and a larger
working set when old versions remain pinned. The logical $\OQ(1)$ bounds in
this paper count dictionary and iterator steps, not cache misses, stalled
cycles, coherence traffic, or remote NUMA accesses.

Dictionary encoding reduces object overhead but does not by itself guarantee
hardware locality. Arrival-order identifiers may be poorly correlated with join
access patterns. The implementation therefore separates generic integer keys
from layout-aware optimizations such as arity-specialized primitive keys, packed
child arrays, packed leaf blocks, local dense remapping, and NUMA-prefix
ownership. Each optimization must improve the measured latency--memory--freshness
frontier against integer-key MVCC and persistent-map baselines.

LFNT makes a different systems trade-off. Published nodes and payloads are
immutable, so a reader does not contend on or invalidate cache lines that update
workers are modifying for a later epoch. This can reduce reader--writer lock
handoffs, ownership transfers, and coherence stalls even when the persistent
layout causes more capacity misses than a flat mutable map. Query factorization,
minimal index selection, prefix sharing, and packed leaf blocks may reduce or
amortize the maintained working set, but no theorem guarantees that LFNT has
better physical locality. Purely functional graph structures similarly require
compressed layouts to make lightweight snapshots practical, while parallel trie
joins remain sensitive to skew and index contention~\cite{dhulipala2019aspen,
wu2025honeycomb}.

We therefore treat locality as a falsifiable architectural hypothesis. The
implementation requirements are stated in Section~\ref{sec:implementation}, and
Section~\ref{sec:evaluation} compares capacity, coherence, TLB, allocation, and
NUMA costs against MVCC, persistent-map, and lock-based alternatives. LFNT is
not retained as the principal mechanism if those costs fail to buy a better
latency--memory--freshness frontier.

\subsection{Versioned Payloads}

A snapshot must freeze multiplicities as well as trie topology. Therefore an
update never overwrites a payload reachable from a published root. It either
copies the changed payload/path or installs an epoch-versioned payload selected
by the snapshot time.

\begin{lemma}[Payload snapshot safety]
\label{lem:payload-safety}
An enumerator reading a published relation root observes the multiplicity
values of that relation version even if later epochs update the same logical
Gmas.
\end{lemma}

\begin{proof}
With copy-on-write, later values are reachable only from later roots. With a
version chain, the pinned snapshot time selects the earlier visible value. In
both cases, later writes cannot change the payload observed through an already
published root.
\end{proof}

\subsection{Atomic Root-Vector Publication}

Let $\mathcal{R}_Q=\{R_1,\ldots,R_m\}$ be the relations and indexes in the
compiled plan. A completed physical snapshot is represented by
\[
  \RootVec_t=\langle t,r^t_{R_1},\ldots,r^t_{R_m}\rangle.
\]
An atomic pointer $\mathsf{PublishedRootVec}$ refers to the latest completed
vector. Update workers construct a working version for epoch $t$; the publisher
installs $\RootVec_t$ with release semantics only after
Algorithm~\ref{alg:process-epoch} finishes. Constructing the vector may touch
$|\mathcal R_Q|$ roots, but this is $\OQ(1)$ because the query is fixed; the
visibility point is one atomic pointer update. An enumerator performs one
acquire load of the pointer and uses only roots from the loaded vector.

Figure~\ref{fig:root-vector-publication} makes the distinction between
relation-level and query-level snapshots concrete. Enumerator $E_1$ acquires
$\RootVec_{t_1}$ and observes the answer count two. Update workers may then
construct and publish $\RootVec_{t_2}$, whose reply and answer counts are one,
without changing any node reachable from the older pinned vector. The
enumerator never assembles a snapshot by reading several current roots
independently.

\begin{figure*}[t]
  \centering
  \resizebox{\textwidth}{!}{\TrieGSSnapshotPublication}
  \caption{Query-level snapshot publication for the running example. A root
  vector groups all relation and index roots for one completed logical time.
  Enumerator $E_1$ remains on $t_1$ while update workers publish $t_2$; this
  prevents a mixture such as $V_1@t_2$, $V_2@t_1$, and $\mathsf{Ans}@t_2$.}
  \label{fig:root-vector-publication}
\end{figure*}

\begin{lemma}[No fractured snapshots]
\label{lem:no-fractured-snapshot}
An enumeration job that loads $\RootVec_t$ observes every maintained relation
and index from the same completed physical snapshot $S_t$, which represents the
same access state $C_t$ and logical view state $L_t$.
\end{lemma}

\begin{proof}
Each component root denotes a relation-level version. The vector is assembled
only after all work for $t$ has completed and is published by one atomic pointer
update. The enumerator reads that pointer once. Lemma~\ref{lem:payload-safety}
prevents later epochs from changing reachable payloads. Hence all reads belong
to $S_t$, and by construction $S_t$ realizes $C_t$ for the completed logical
state $L_t$.
\end{proof}

\paragraph*{Linearization points.}
Atomic publication of $\RootVec_t$ is the visibility point of update epoch $t$.
The enumerator's atomic load of $\mathsf{PublishedRootVec}$ is its snapshot
linearization point. If an update is not yet published, the job legitimately
enumerates an older state rather than waiting for the update.

\subsection{Progress and Reclamation}

Individual LFNT updates inherit lock-free progress from the selected concurrent
trie. The complete maintenance pipeline may include queues and an ordered epoch
publication barrier, so we do not claim that an entire epoch is wait-free.
Enumeration is non-blocking with respect to maintenance: after loading a vector,
it traverses immutable versions without locking update paths.

Relation-level compaction can remove zero-payload keys, but query-level
reclamation is stricter. Old roots and payloads are reclaimed only after no
active enumerator can reach them. Epoch-based reclamation, hazard pointers, or
reference counting are all compatible with the proof. Slow readers can therefore
create reclamation debt; this is an explicit cost, not hidden by the asymptotic
delay theorem.

\subsection{Main Result}

\begin{theorem}[Snapshot-consistent \TrieGS guarantee]
\label{thm:core}
Let $Q$ be a fixed full acyclic conjunctive graph query evaluated under
multiset semantics over a totally ordered update stream. Under serializable
epoch execution, the relation-level LFNT and indexed-traversal contracts,
snapshot-preserving payloads, atomic root-vector publication, and safe
reclamation, every enumeration job $e$ observes a snapshot for some published
time $\tau(e)$ and emits every nonzero pair
$\bigl(\bar a,\mu_{\tau(e)}^Q(\bar a)\bigr)$ exactly once with $\OQ(1)$ delay.
Later update epochs may proceed concurrently without modifying the pinned
snapshot.
\end{theorem}

\begin{proof}
Theorem~\ref{thm:abstract-cde-enumeration} gives the abstract constant-delay
support enumerator for the fixed decomposition. Theorem~\ref{thm:dynamic-multiset-maintenance}
restores the counting view invariant at every completed epoch.
Proposition~\ref{prop:access-accounting} and Assumption~\ref{ass:indexed-extension-access}
realize the required local accesses with $\OQ(1)$ delay in the fully indexed
configuration. Lemma~\ref{lem:gma-instantiation} instantiates each exact-key
relation operation with the imported C-trie/cache-trie guarantees.
Lemma~\ref{lem:no-fractured-snapshot} shows that an enumerator observes one
completed physical snapshot of this access state. Lemma~\ref{lem:logical-enumeration-correctness}
gives weighted-answer correctness, and Lemma~\ref{lem:logical-constant-delay}
gives the delay bound. Enumeration does not acquire locks on mutable update
paths, so later maintenance can continue; safe reclamation retains old versions
until the job unpins them.
\end{proof}

\begin{corollary}[Projected free-connex queries]
\label{cor:projected-materialized}
Let $Q$ be a fixed projected free-connex acyclic query for which \TrieGS
materializes the count-aware relation $\mathsf{Ans}_Q$. Every enumeration job
has the same snapshot-consistency guarantee and scans the nonzero answer pairs
with $\OQ(1)$ delay. The update and space bounds additionally include the
changed and stored tuples of $\mathsf{Ans}_Q$.
\end{corollary}

\begin{proof}
The update proof treats $\mathsf{Ans}_Q$ as the final count-aware projection.
The view invariant therefore gives each projected tuple its exact multiset
count. Root-vector publication snapshots the answer relation with the remaining
plan, and its support index advances directly between nonzero tuples.
\end{proof}

\begin{corollary}[Sliding-window queries]
For either construction above over active window state $G_t^W$, the
corresponding conclusion holds after replacing $G_t$ by $G_t^W$.
\end{corollary}

\begin{proof}
The window manager reduces arrivals and expirations to multiset insertions and
deletions. No other proof obligation changes.
\end{proof}

\paragraph*{Arithmetic model.}
The operation bounds assume counts fit in a machine word. Arbitrary-precision
multiplicities add a cost depending on their bit length. Silent integer
overflow violates the semantic theorem.

\section{Implementation and Practical Extensions}
\label{sec:implementation}

\subsection{Execution Architecture}

The runtime has four roles. Update workers process affected Gmas in a work
queue; a publisher installs completed root vectors in logical-time order;
enumeration workers pin published vectors; and a reclaimer frees versions no
longer reachable by any reader. Work stealing may balance independent update
branches, following the general principle of fine-grained multicore scheduling
used in parallel query engines~\cite{DBLP:conf/sigmod/LeisBK014}. Scheduling is
an implementation choice rather than part of the CDE theorem.

For self-joins or replicated access orders, each base update carries one epoch
identifier to all logical aliases. In the running example, a
$\mathsf{hasCreator}$ update is routed to both
$\mathsf{hasCreator}_v$ and $\mathsf{hasCreator}_z$. Publication waits until
every alias, projection, secondary index, and dependent view has processed that
epoch. Later epochs may be prepared in a pipeline, but cannot be published
before their dependencies.

\paragraph*{Claim boundaries.}
The implementation separates four levels of claims.
\begin{enumerate}[leftmargin=*]
  \item \emph{Epoch-versioned root publication} is the theorem mechanism.
  Update workers construct later versions without modifying nodes reachable
  from a published state. Once epoch $t$ completes, the publisher installs one
  immutable root vector with a single atomic release store; no hardware
  transaction is required.

  \item \emph{Dictionary-coded access indexes} are a semantic-preserving physical
  refinement. RDF terms and relation symbols are mapped to stable integer
  identifiers, so hot paths use primitive fixed-arity keys. The theorem assumes
  only injectivity; locality-aware identifier assignment is an evaluated
  optimization.

  \item \emph{Packed LFNT indexes} are a physical design. A fixed query gives
  query-bounded logical key depth, while the selected cache-trie implementation
  supplies its expected or amortized exact-key operation bounds. This does not
  imply one-memory-access or one-CAS traversal, and cache or NUMA locality is
  evaluated rather than assumed.

  \item \emph{Shared-prefix maintenance} is an optional extension. Exactly
  equivalent canonical subplans may share one maintained view, eliminating
  duplicate logical maintenance. Any reduction in cache misses, branch
  mispredictions, coherence traffic, or remote-memory accesses is an empirical
  hypothesis tested against unshared and non-coalesced alternatives.
\end{enumerate}

\subsection{Morselized Epoch Scheduling}
\label{sec:morselized-scheduling}

TrieGS adapts the scheduling principle of morsel-driven parallelism~\cite{DBLP:conf/sigmod/LeisBK014} to dynamic
maintenance. A morsel-driven analytical engine schedules small fragments of a
one-shot query pipeline; TrieGS instead schedules affected maintenance fragments
inside one update epoch. The semantic boundary is not a pipeline breaker, but
root-vector publication.

\paragraph*{Maintenance morsels.}
For epoch $t$, the runtime decomposes propagation into small tasks tagged with
an epoch, a relation or rule, and a high-level LFNT prefix. We use four morsel
types.
\begin{enumerate}[leftmargin=*]
  \item \emph{Base morsels} apply a batch of base Gma updates for one relation
  and prefix.
  \item \emph{Projection morsels} update count-aware parental projections for
  changed child Gmas.
  \item \emph{Rule morsels} evaluate affected bindings for one guard-normal rule
  and one guard or interface prefix.
  \item \emph{Answer morsels} update an explicit $\mathsf{Ans}_Q$ relation or
  expose a factorized full-query answer generator.
\end{enumerate}
In the running example, the insertion of
$\mathsf{likedBy}(\mathsf{tweet2},\mathsf{Mia})$ creates a base morsel for the
$\mathsf{likedBy}$ prefix, a projection morsel for
$L^p(\mathsf{tweet2})$, a rule morsel for
$V_2(\mathsf{tweet2},\mathsf{Eva})$, a projection morsel for
$V_2^p(\mathsf{tweet2})$, and finally a rule morsel for
$V_0(\mathsf{tweet2})$.

\paragraph*{Scheduler invariants.}
The scheduler may reorder and parallelize morsels inside an epoch only when the
following invariants hold.
\begin{description}[leftmargin=*]
  \item[Epoch isolation.] Every morsel writes only to the working version for its
  logical time $t$.
  \item[Dependency closure.] A morsel that changes a Gma enqueues every
  dependent projection and rule morsel required by the compiled view-dependency
  graph.
  \item[Exact final value.] Conflicting writes to the same derived Gma are
  serialized or combined so that the value installed before publication equals
  the deterministic \textsc{ProcessEpoch} result.
  \item[Publication barrier.] $\RootVec_t$ is published only after the
  outstanding-morsel count for epoch $t$ reaches zero and all writes to the
  working version are visible.
  \item[Snapshot immutability.] No morsel modifies a node or payload reachable
  from any previously published root vector.
\end{description}

\begin{lemma}[Morselized epoch refinement]
\label{lem:morselized-refinement}
If the morsel scheduler satisfies the invariants above, then a morselized
execution of epoch $t$ publishes the same logical state $L_t$ as
Algorithm~\ref{alg:process-epoch}.
\end{lemma}

\begin{proof}
Every input update is represented by a base morsel. By dependency closure, every
changed Gma schedules the projection and rule work that
Algorithm~\ref{alg:process-epoch} would enqueue. By exact-final-value
serialization or combination, each derived Gma installed before publication is
the deterministic value computed by \textsc{SetExact}. Since publication waits
until all morsels complete, no dependent change is omitted from $\RootVec_t$.
Snapshot immutability confines all writes to the working version. Thus the
published state refines the serializable epoch semantics used in
Lemma~\ref{lem:update-correctness}.
\end{proof}

\paragraph*{NUMA-aware prefix ownership.}
On the Intel NUMA platform, high-level integer LFNT prefixes may be assigned
owner sockets and allocated by first touch. Workers first consume morsels for
local prefixes; remote stealing is allowed only when local queues fall below a
threshold or when imbalance would otherwise delay the publication barrier. This
policy is evaluated against serial, global-FIFO, ordinary work-stealing,
prefix-coalesced, and adaptive schedulers; it is not a correctness requirement.

\subsection{CQELS SWAG Integration}
\label{sec:swag-integration}

CQELS already contains SWAG-style sliding-window aggregation operators for
aggregate summaries over FIFO windows~\cite{tangwongsan2020daba,tangwongsan2023bulk}.
TrieGS connects to this layer in two distinct ways.

\paragraph*{Base-window delta source.}
For the core query semantics, the window is over base graph events. CQELS
windowing determines which events enter and leave the active window, and TrieGS
receives those changes as ordinary insertion and expiration epochs. The core
sliding-window corollary therefore remains a statement about $Q(G_t^W)$, not a
statement about aggregate summaries.

\paragraph*{Answer-delta aggregate consumer.}
For aggregate continuous queries, TrieGS can emit exact signed answer deltas
\[
  (\bar a,\delta\mu_t^Q(\bar a))
\]
from the maintained answer representation. A SWAG aggregate adapter maps each
answer delta to a group key and aggregate contribution over integer dictionary
IDs. If the aggregate result must be enumerated with the same snapshot as the
answers, its aggregate roots are included in $\RootVec_t$; otherwise SWAG is a
downstream consumer with its own lag and consistency contract.

\begin{proposition}[SWAG aggregate extension]
\label{prop:swag-extension}
Let $\gamma$ be an aggregate whose state can be updated from weighted answer
deltas. If TrieGS emits the exact signed delta of $\mathsf{Ans}_Q$ for every
published epoch and the SWAG aggregate state applies all deltas through epoch
$t$, then the aggregate state equals applying $\gamma$ to the declared answer
multiset or answer-event window at $t$.
\end{proposition}

\begin{proof}
The TrieGS answer-delta stream is exact by the view invariant and
\textsc{SetExact}: every answer payload change is represented by its signed
difference. Applying a delta-compatible aggregate update to every signed answer
contribution therefore produces the same aggregate as recomputing the aggregate
from the target answer multiset. The timestamp or root-vector inclusion decides
which snapshot the aggregate corresponds to.
\end{proof}

For \texttt{COUNT} and \texttt{SUM}, signed deltas update the aggregate directly;
\texttt{AVG} maintains weighted sum and count. For \texttt{MIN} and
\texttt{MAX}, deleting the current extremum requires additional support
structures unless the declared semantics is a FIFO answer-event window. This is
an aggregate-extension obligation, not part of Theorem~\ref{thm:core}.

\subsection{Dictionary-Coded Physical Design}

The logical LFNT does not require one Java object per trie word. The 
backend should use dictionary IDs from Section~\ref{sec:dictionary-encoding} and
avoid the object-key layout used by the correctness prototype. Relation IDs,
view IDs, predicate IDs, variable values, and SWAG group keys are represented as
primitive integers. Binary base edges and low-arity views should use
arity-specialized primitive keys; larger views use packed primitive arrays or
leaf blocks.

Packed LFNT nodes store child labels and child references in compact arrays or
bitmap/dense-node formats instead of comparator-driven object maps. Several
nearby leaf Gmas may share one contiguous block containing integer suffixes and
long multiplicity payloads. Packing amortizes a fetched cache line across
several keys and reduces object-header, comparator, TLB, and garbage-collection
overhead. The root vector is likewise represented as an indexed root array on
the hot path rather than a hash map lookup per relation.

All indexes required by propagation and enumeration are selected at compile
time and constructed before publication. An enumerator never lazily builds an
index, installs a shortcut, updates access statistics, or otherwise writes into
a published node. Mutable task counters, pin counts, metrics, and reclamation
metadata are kept outside immutable query nodes to avoid false sharing with
reader data.

\subsection{Space Optimizations}

\paragraph*{Parental-projection packing.}
A projection Gma can point to the factorized child extension set that justifies
its witness count rather than duplicating all child keys. The count remains an
independent logical payload so deletions are correct.

\paragraph*{Guard-child packing.}
When a view and its selected guard have the same key variables, their trie paths
can share prefixes. This reduces key and node duplication but does not merge
their logical multiplicities.

\paragraph*{Index packing.}
Several attribute orders may share common trie prefixes. The optimizer selects
a minimal set of physical indexes that covers update propagation and
enumeration requirements. Every packed layout must preserve the same logical
$\scanIndex$ interface assumed by the proof.

\paragraph*{Fractional products.}
For a grounded product view with many body factors, the implementation may
retain the product of nonzero factors and a zero-factor count. This permits a
factor update without re-reading all siblings. The optimization is valid only
when division is exact or when the old factor is retained; transitions through
zero must update the zero count explicitly.

\subsection{Multi-Query Sharing}
\label{sec:multiquery}

Queries are first compiled independently into canonical logical subplans.
Exactly equivalent subplans may share one maintained view and its indexes,
with reference counts controlling reclamation. This exact-sharing layer is the
formal baseline. Frequent-pattern mining, such as GSpan-style discovery of
approximately reusable subpatterns~\cite{DBLP:conf/icdm/YanH02}, is optional and
must account for different projections, filters, access orders, and output
requirements. Classical multi-query optimization motivates the sharing
objective~\cite{sellis1988multiple,roy2000efficient,mistry2001materialized}.

\subsection{Features Outside the Core Theorem}

\paragraph*{Filters and constants.}
Label, type, equality, and indexed value predicates remain within the theorem
when each check has $\OQ(1)$ cost. More general predicates add their evaluation
cost to the delay.

\paragraph*{Aggregation.}
\texttt{COUNT} sums multiplicities; \texttt{SUM} weights values by
multiplicity; \texttt{AVG} retains weighted sum and count; \texttt{MIN} and
\texttt{MAX} require support structures for deletions; and
\texttt{COUNT DISTINCT} counts values with positive support. These operators
use the snapshot protocol but require aggregate-specific maintenance proofs.
SWAG integration supplies the CQELS-native implementation path for aggregate
summaries, but it does not replace the full weighted answer iterator.

\paragraph*{Cyclic queries.}
A bounded-width cyclic subquery may be encapsulated as a materialized bag whose
interface joins an acyclic skeleton. Specialized triangle maintenance,
factorized IVM, or a WCOJ evaluator can implement the bag~\cite{DBLP:journals/tods/KaraNNOZ20,
nikolic2018incremental,DBLP:journals/pacmmod/WangWS23}. The total cost is the
bag maintenance cost plus the \TrieGS skeleton cost; Theorem~\ref{thm:core}
does not imply a constant-delay or update bound for an arbitrary bag.

\paragraph*{SPARQL and paths.}
The positive basic-graph-pattern fragment maps naturally to conjunctive-query
bag semantics. \texttt{OPTIONAL}, \texttt{MINUS}, correlated
\texttt{NOT EXISTS}, ordering, arbitrary property paths, and entailment regimes
require distinct semantics. Streaming regular path queries are therefore an
extension rather than a consequence of the core theorem~\cite{Pacaci:2020,
DBLP:conf/icde/PacaciBO22}.

\subsection{Durability Boundary}

A linearizable in-memory root vector is not a durable checkpoint. Exactly-once
recovery must atomically bind a published vector to the committed input offset
or watermark and to the sliding-window expiration queue. The current theorem
covers in-memory concurrency; Section~\ref{sec:evaluation} includes failure
injection to determine the engineering cost of extending the protocol to
crash-consistent operation.

\section{Evaluation}
\label{sec:evaluation}

This section reports a first-round evaluation of the current TrieGS/CQELS
artifact. 
The current data support four main observations. First, the implementation now
passes meaningful semantic gates, including randomized update streams,
publication snapshots, and configured TrieGS/Flink parity. Second, the workload
is no longer a static or binary-join smoke test: it includes 11 SNB-shaped
sliding-window query templates and post/comment temporal joins. Third, the
current relation-backend matrix is negative for LFNT performance: simpler
sequential and MVCC-object backends dominate the configured prototype series.
Fourth, CPU placement is a first-order scaling variable on the NUMA host:
CPU-node affinity improves every SNB-shaped query, even though memory binding was
not available.

\subsection{Experiment Steup}
\label{sec:eval-artifact}
We setup two multicore environments: Uniform Memory Access (UMA) and Non-uniform memory access (NUMA)
The CQELS/Flink runs were executed on a UMA host ( Apple M2 Ultra with 24
physical/logical cores, macOS 15.3.2 ARM64, 196GB RAM and OpenJDK 24). The SNB-backed
profiles use the SF0.1 Turtle dataset. 
The NUMA-node  analysis uses a configuration of a server with 4 NUMA Intel CPUs with 112 cores, 2TB RAM.
\subsection{Correctness Tests}
\label{sec:eval-correctness}

TrieGS maintains weighted answers through base relations, projections, derived
views, and published root vectors. A throughput result is meaningless if any one
of those layers loses multiplicity information or exposes a fractured state. The
first evaluation question is therefore semantic: does the implementation agree
with an oracle and with an independent streaming execution path on the configured
workloads?

The randomized semantic gate uses 100 fixed seeds and 10{,}000 epochs per seed.
The SNB publication gate accepts 3{,}056 SNB Turtle edges, publishes 13 epochs
and 13 snapshots, and reports zero impossible snapshots. These numbers exercise
the two central proof obligations: update epochs must restore the view invariant,
and readers must observe one published state rather than a mixture of relation
roots.

\nop{
The bounded SNB person-name BGP remains a smoke gate rather than the main parity
claim:
\begin{lstlisting}[language=SQL]
?person snvoc:firstName ?firstName .
?person snvoc:lastName  ?lastName .
\end{lstlisting}
}
On the normalized workload, the Flink DataStream baseline produces the same
1{,}528 final answers, the same total multiplicity, and identical answer
checksums.

\subsection{UMA Scaling: Parallelism Exists, But Saturates Early}
\label{sec:eval-m2}

The UMA scaling profile runs the 11 SNB-shaped query templates at
\texttt{eventCount=760}, with worker counts 1, 2, 4, 8, 16, and 24. Geometric
mean throughput improves through 8 workers, then plateaus at the 16
performance-core and 24 all-core points. The best current geometric-mean speedup
is 3.28x at 8 workers.

\begin{table}[t]
\centering
\small
\caption{UMA scaling summary over 11 SNB-shaped queries.}
\label{tab:m2-scaling}
\begin{tabular}{rrr}
\toprule
Workers & Geomean ops/s & Speedup \\
\midrule
1 & 1{,}932.32 & 1.00x \\
2 & 3{,}163.32 & 1.64x \\
4 & 5{,}240.92 & 2.71x \\
8 & 6{,}333.32 & 3.28x \\
16 & 6{,}199.97 & 3.21x \\
24 & 6{,}218.92 & 3.22x \\
\bottomrule
\end{tabular}
\end{table}

\paragraph*{Insight.}
It
shows that the current worker model can extract useful parallelism, but that the
parallel granularity or backend path saturates before the 16-core target. This
points to likely bottlenecks: publication barriers, shared queues, limited
per-epoch affected work, allocation pressure, or relation-backend overhead. The
UMA result is therefore a guide for the NUMA affinity and profiling experiments,
not a completed multicore claim.

\paragraph*{Engine-path morsel scaling.}
The UMA numbers above measure the SNB harness, which parallelizes a fixed set of
sliding windows with a \texttt{parallelStream}; their plateau is a task-count
ceiling of the harness rather than a property of the engine's intra-epoch morsel
parallelism. To measure the engine path directly we added a steady-state
benchmark that drives the incremental maintenance operator over one advancing
window---a FIFO residency of node-disjoint edge blocks, replayed periodically so
per-slide work is stationary and bounded---applying each slide through the
maintenance morsel pool so the parallel width is the morsel count, far above the
core count. The first measurement on this path showed \emph{no} multicore
speedup: per-epoch cost was dominated by sequential $O(\text{state})$ snapshot
bookkeeping (copying the base graph state and rebuilding every counting view each
epoch), and the workload was allocation-bound---profiling attributed the bulk of
a slide's allocation to a per-comparison \texttt{TreeSet} inside the
answer-binding comparator and to a \texttt{Binding} object materialized for every
enumerated tuple.

We then rebuilt the per-epoch maintenance to be $O(\Delta)$ end to end and to
represent all hot-path state columnarly. The rewrite proceeded as a sequence of
individually measured, oracle-gated steps: dictionary-encoded, dense-integer
columnar counting views and candidate indices, so an enumerated tuple is a tuple
of \texttt{int}s and a \texttt{Binding} is materialized only at the API boundary;
an allocation-free binding comparator; a slot-indexed integer valuation in the
enumeration inner loop; a sharded delta accumulator whose per-worker partials are
combined shard-parallel, eliminating the single-threaded final reduction of the
tree-structured combiner; derivation of each epoch's changed edges by one batch
scan rather than per-edge index probes; a \emph{mutable occurrence-count ledger}
that replaces the per-epoch persistent base-state copy (the copy existed only to
carry a base-count snapshot that no hot-path consumer reads); and overlapping the
signed answer-delta decode with the parallel view merges. Each step is gated by a
differential oracle against brute-force enumeration over both multigraph and set
semantics and by a parallel-equals-sequential parity test.

Together these raise single-thread throughput several-fold over the already
order-of-magnitude-improved $O(\Delta)$ baseline and lift the intra-epoch
multicore speedup on the steady-state engine benchmark from the roughly
$1.3\times$ of the view-parallel-only version to about $1.9\times$ at eight
workers (measured on a $24$-core host; the speedup is stable across join fan-out).
Two effects compound: the columnar rewrite removes the per-tuple allocation and
per-comparison string work that had left the parallel section garbage-collector
bound, and the ledger both eliminates the largest remaining sequential phase and,
by cutting roughly $140$k node allocations per slide, relieves shared-allocator
pressure on the parallel sections---so cheapening a sequential phase also improves
parallel scaling.

The $\approx 1.9\times$ figure is where the thin-query story ends, not where
the scaling story ends. Section~\ref{sec:eval-anatomy} decomposes where the
remaining time goes, shows that the intra-epoch ceiling is a property of the
plan's compute density rather than of the engine, and reports the negative
results and ceiling probes that pin the residual bound to structural
parallel width.
Section~\ref{sec:eval-multiquery} then shows that the pre-registered
$\ge 8\times$ target is met at inter-query granularity, against a
work-matched Flink baseline. The root-vector publication protocol remains the
reusable systems contribution.

\subsection{Scaling Anatomy of the Maintenance Operator}
\label{sec:eval-anatomy}

The engine-path benchmark (Section~\ref{sec:eval-m2}) left the thin 3-atom
path at $\approx 1.9\times$ on eight workers. Two questions remain: is that
ceiling a property of the operator or of the workload, and was the serial
fraction actually engineered out before blaming memory bandwidth? This
subsection answers both with a plan-shape sweep and a phase-level Amdahl
decomposition~\cite{Amdahl:1967:VSP:1465482.1465560}.

\paragraph*{Speedup tracks compute per resident byte.}
Table~\ref{tab:intra-query-shapes} sweeps morsel parallelism (one query, one
shared state, a worker pool) across plan shapes of increasing enumeration
density. The compute-dense shapes --- a 4-armed spider with branching 3 and a
5-way star with branching 6 --- reach $7.5$--$7.8\times$ at 16 workers, the
performance-core budget of the host. The negative control, a 5-atom path with
prefix projection whose per-seed enumeration is shallow, stays at
$1.3\times$ at every worker count; the thin 3-atom path of the steady-state
benchmark measures $\approx 1.9\times$. The mechanism: per-seed enumeration
work grows multiplicatively with join branching, while resident state --- and
with it the serial per-epoch cost of sealing indexes and merging views ---
grows only linearly. Dense plans therefore decouple parallel compute from
serial bookkeeping; thin paths never do, because each morsel touches barely
more bytes than the seal itself must. Intra-query speedup is a function of
compute per resident byte, not a global constant of the engine --- and the
spread from $1.3\times$ to $7.8\times$ across shapes confirms the mechanism
rather than a cherry-picked workload.

\begin{table}[t]
\centering
\small
\caption{Intra-query morsel scaling by plan shape (steady state, one query,
speedup vs.\ one worker; $t{=}1$ in slides/s).}
\label{tab:intra-query-shapes}
\begin{tabular}{lrrrr}
\toprule
Shape & $t{=}1$ & $t{=}8$ & $t{=}16$ & $t{=}24$ \\
\midrule
\textsc{Spider4x2} (branching 3) & 22.5 & 4.76x & \textbf{7.53x} & 5.97x \\
\textsc{Star5} (branching 6) & 29.7 & 5.02x & \textbf{7.82x} & 7.60x \\
\textsc{Path5-prefix} (neg.\ control) & 104.3 & 1.28x & 1.29x & 1.25x \\
Thin 3-atom path (Sec.~\ref{sec:eval-m2}) & --- & $\approx$1.9x & --- & --- \\
\bottomrule
\end{tabular}
\end{table}

Flink has no counterpart for Table~\ref{tab:intra-query-shapes}: a single
continuous query of this kind is one key in a keyed stream (or a non-keyed
window), so its job shape is structurally parallelism-1 per query regardless
of the cluster's parallelism setting. Intra-query morsel parallelism is an
architectural differentiator, not a tuning difference.

\paragraph*{The serial fraction was engineered down step by step.}
The $1.9\times$ thin-path figure is the end of a measured ladder, not a first
attempt. Starting from the $O(\Delta)$ indexed operator, five individually
oracle-gated steps moved the eight-worker/one-worker ratio:
columnar views and indices, $1.26\times$; a sharded delta combiner replacing
the single-threaded tree reduction, $1.40\times$; single-scan derivation of
each epoch's changed candidates, $1.53\times$; a mutable occurrence-count
ledger replacing the per-epoch persistent base-state copy, $1.77\times$; and
overlapping the answer decode with the parallel view merges, $1.88\times$.
Absolute throughput moved more than the ratio: at eight workers the thin
fan-out-8 configuration went from 26.6 to 183 slides/s ($+588\%$), and even
single-threaded from 18 to 98 slides/s ($+444\%$); per-slide allocation fell
from 41--45\,MB to 23--27\,MB ($-44\%$).

\paragraph*{Phase decomposition at the end of the ladder.}
At eight workers the fan-out-8 slide takes 5.46\,ms wall: validation
(0.33\,ms) and index advancement (1.11\,ms) are serial but now jointly a
quarter of the slide; delta computation (1.55\,ms) speeds up $2.1\times$ on
eight workers and the view merge plus decode (2.25\,ms) $2.37\times$. These
four numbers reproduce the observed end-to-end ratio ($1.44 + 1.55{\cdot}2.1
+ 2.25{\cdot}2.37 \approx 10.0$\,ms at one worker, $10.0/5.46 = 1.84\times$),
an internal-consistency check on the decomposition. The ceiling is therefore
no longer Amdahl's law over an unexamined serial tail: the serial phases are
small, and the bound is the $\approx 2\times$-on-eight-workers scaling of the
two \emph{parallel} phases, whose per-morsel work on a thin plan is
essentially a memory traversal.

\paragraph*{Negative results: the ceiling was pushed, not assumed.}
Four further parallelization attempts were built, oracle-gated, and measured
net-negative at the per-epoch delta sizes of this workload:
hash-partitioning the largest counting view; range-partitioning the columnar
index seal (built byte-identical to the serial seal and gated before
measurement); batching the base-state application; and contiguous-block
morsel striping, which improves cache locality but regresses throughput
because it sacrifices the load balance of round-robin assignment across
atom-positioned morsels. In each case fork and coordination overhead exceeded
the gain. We report these alongside the positive ladder deliberately: the
memory-bandwidth attribution for the residual ceiling is credible precisely
because the alternative explanations --- residual serial work, poor layout,
unlucky partitioning --- were each implemented and refuted rather than argued
away.

\paragraph*{A structural-width ceiling, and how we know.}
The negative results above left memory bandwidth as the most credible
residual explanation. Three further probes replace it with a sharper one:
the binding constraint is \emph{structural parallel width}, not bandwidth.
First, two allocation-reduction slices --- reusing per-slide scratch buffers
and parallelizing index advancement --- cut per-slide allocation from 20.1
to 14.1\,MB ($-30\%$) on a refreshed fan-out-8 baseline and lift
single-thread throughput by $+9.9\%$, yet move no speedup ratio in the
useful direction: both thread counts gain, and the eight-worker/one-worker
ratio in fact slips from $1.86\times$ to $1.72\times$ because the
single-thread side gains more. The methodological corollary is blunt:
absolute-cost cuts, however large, cannot chase a strong-scaling target ---
a broadly uniform improvement leaves the ratio where it was.

Second, we attacked the remaining serial structure on the enumeration's
output side directly: a flat per-worker append-log for enumerated deltas
with a views$\times$shards parallel aggregation. The implementation was
completed and oracle-gated correct --- and measured \emph{net-negative},
$-5\%$ at eight workers, because each aggregation shard must rescan every
worker's log for its keys and that per-shard rescan redundancy exceeds what
lock-free insertion saves. It was reverted. As with the earlier negative
results, at microsecond task grain the coordination is the workload.

Third, and decisively, a compute-density invariance probe: raising join
fan-out from 8 (64 completions per edge) to 16 (256 per edge) quadruples
enumeration density per delta, yet both configurations plateau at
$\approx 1.64\times$ by four workers and stay flat through 16 --- identical
curves at $4\times$ different density. A delta-size or bandwidth explanation
predicts the curves separate; they do not. Nor is DRAM the binding resource:
the streaming seal-merge phase of the same operator scales $3.5\times$. What
remains is structural: a fine-grained slide decomposes into a small fixed
number of independently updatable structures --- three counting views, three
relations, one output view --- and no scheduler can extract more parallelism
than that width supplies. The honest conclusion is that no tractable change
lifts fine-grained per-slide maintenance past $\approx 1.64\times$ on this
path. The one identified lever is partitioning the output view itself, so
aggregation width scales with data rather than with plan structure; that
requires a chunked copy-on-write columnar redesign whose merge-only cousin
already measured net-negative above, and we leave it as future work
(Section~\ref{sec:eval-gaps}).

\paragraph*{A second operator hits the same ceiling.}
The width ceiling is not specific to counting-view maintenance. The SWAG Tier-2
batch-recompute operator --- a different intra-query parallelization, splitting
one aggregate recomputation across a \texttt{ForkJoinPool}, oracle-checked
against a closed form and an independent F-IVM baseline --- plateaus at
$\approx 3.4\times$ (small state) to $\approx 2.5\times$ (large state) on the UMA
and barely beats a no-split sequential baseline at $n{=}10^6$ ($1.27\times$). On
the 4-socket Intel host it reaches $\approx 4.4\times$ \emph{within one
memory-local socket} but is \emph{slower} across sockets than on 14 local cores
($0.7$ vs.\ $0.9$ at $n{=}10^6$): the single shared pooled computation pays
remote-memory cost past one NUMA node. A second, independently-oracle-checked
operator therefore reaches the same low intra-query ceiling, and adds a
locality corollary --- intra-operator parallelism is bounded not just by
structural width but by the single memory domain it can profitably use.

\paragraph*{Insight.}
Speedup ratios are a misleading engineering target. Two of the columnar steps
\emph{lowered} the eight-worker/one-worker ratio while raising absolute
throughput at both thread counts: they cheapened the phases that scale, which
shrinks the parallel fraction --- Amdahl's speedup metric punishes exactly
the optimizations that make the parallel section fast. The ledger step did
the opposite and is the more interesting datum: replacing the persistent
base-state copy removed a serial phase \emph{and} improved the scaling of
the untouched parallel phases, lifting eight-worker throughput by more than
one-worker throughput. The mechanism is that the garbage collector is a
shared resource: the copy allocated $\sim$140{,}000 nodes per slide, and its
allocation and collection pressure was paid disproportionately by the phases
running on eight threads. Cheapening a sequential phase improved parallel
scaling. Neither effect is visible in a speedup ratio alone; we therefore
report absolute throughput at fixed thread counts as the primary metric and
treat ratios as diagnostics.

\subsection{Multi-Query Scaling and a Work-Matched Flink Baseline}
\label{sec:eval-multiquery}

The engine-path study bounds what parallelism can do \emph{inside} one epoch
of one thin query. The classic streaming-engine scalability claim, however,
lives at a different granularity: many continuous queries sharing one host.
This subsection measures that regime with a fixed-work strong-scaling
protocol --- 24 independent, namespaced query lanes, partitioned across $T$
driver threads by ownership, identical total work at every $T$ --- and pairs
it with a work-matched Apache Flink~1.18.1 baseline on the same host (24
cores, 16 performance + 8 efficiency; local streaming MiniCluster, heap
state backend, object reuse on, no checkpointing).

Both systems consume the \emph{same} slide schedules from a single workload
generator, and lanes map to Flink subtasks with the same ownership
quantization as the TrieGS driver round-robin (balanced key assignment;
natural key hashing over 24 keys is skewed and would handicap Flink). The
Flink operator (variant A) is a textbook incremental multi-way join in a
\texttt{KeyedProcessFunction} --- per-relation multiset adjacency indexes in
keyed state, $O(\Delta)$ work per slide --- i.e.\ work-model-matched to
TrieGS and written as a competent Flink user would write it, since Flink
ships no incremental multi-way join. An idiomatic variant~B (native sliding
event-time windows, full recompute per fire) is implemented and
parity-gated as a bracket; it is strictly slower per slide by construction.
Every number sits behind a semantics gate: the Flink job's per-slide signed
answer-delta multisets are asserted equal to the brute-force oracle at
parallelism 1 and 2, and within each sweep the per-mix answer checksum is
bit-identical across every parallelism level --- a work-identity control
that rules out the parallel configurations silently doing different work.

Table~\ref{tab:multiquery-headline} reports the headline numbers; the mixes
are 24 thin 3-atom path lanes (\textsc{Path3}), 24 compute-dense 4-armed
spider lanes (\textsc{Spider4x2}), and a heterogeneous third-third-third mix
(\textsc{Mix3}: 8 path + 8 spider + 8 five-atom star lanes). Because the
host is heterogeneous (efficiency cores measure $\approx 0.55\times$ a
performance core on this workload), the relevant ideal is
$\approx 20.4\times$, not $24\times$; we report efficiency against that
figure.

\begin{table*}[t]
\centering
\small
\caption{Multi-query strong scaling, TrieGS vs.\ work-matched Flink
(24 lanes, fixed work, aggregate slides/s; speedup vs.\ each system's own
$T{=}1$; heterogeneous ideal $\approx$20.4x on this 16P+8E host). Last row:
intra-query morsel parallelism, which has no Flink counterpart (one
continuous query is one key, hence structurally parallelism-1).}
\label{tab:multiquery-headline}
\begin{tabular}{llrrrr}
\toprule
Mix & System & $T{=}1$ & $T{=}24$ & Speedup & Abs.\ gap \\
\midrule
\textsc{Path3} & TrieGS & 270.8 & 4{,}426.5 & 16.35x & \multirow{2}{*}{5.8--6.2x} \\
 & Flink & 46.8 & 718.0 & 15.36x & \\
\addlinespace
\textsc{Spider4x2} & TrieGS & 26.6 & 471.2 & 17.70x & \multirow{2}{*}{29--32x} \\
 & Flink & 0.83 & 16.4 & 19.84x & \\
\addlinespace
\textsc{Mix3} & TrieGS & 43.9 & 574.2 & 13.08x & \multirow{2}{*}{17--30x} \\
 & Flink & 2.57 & 19.4 & 7.57x$^{\dagger}$ & \\
\midrule
\multicolumn{2}{l}{Intra-query (\textsc{Star5}, morsels)} & 29.7 &
\multicolumn{2}{r}{7.82x at $t{=}16$} & no counterpart \\
\bottomrule
\multicolumn{6}{l}{\footnotesize $^{\dagger}$Plateau from $T{=}8$
(7.68x): slowest-class gating, see below.}
\end{tabular}
\end{table*}

Three observations. First, \emph{on homogeneous mixes the two systems scale
near-identically}: on the thin path, TrieGS reaches $16.3\times$ and Flink
$15.4\times$ at $T{=}24$ (80\% and 75\% of the heterogeneous ideal), with the
curves tracking within a few percent at every intermediate point, including
the same dip at $T{=}16$ (ownership quantization, below). On the
compute-dense spider Flink is in fact slightly steeper ($19.8\times$ vs.\
$17.7\times$) because its far higher per-slide cost amortizes framework
overhead. TrieGS's inter-query scaling is framework-grade; nothing in the
engine's shared infrastructure impedes it.

Second, \emph{the absolute gap is large and parallelism-independent}: TrieGS
is $5.8$--$6.2\times$ faster on the thin path and $\sim$29--32$\times$ faster
on the compute-dense spider at every thread count, under a work-equivalent
operator and identical schedules. Since the scaling curves match, the gap is
per-slide engine cost --- columnar $O(\Delta)$ maintenance with no
serialization or mailbox machinery --- not a parallelization artifact.

Third, on the heterogeneous \textsc{Mix3} the systems \emph{diverge in
scaling, not just in level}: Flink plateaus at $7.6\times$ from $T{=}8$
onward while TrieGS continues to $13.1\times$ at $T{=}24$. This is the one
qualitative difference in the sweep, and it has a precise mechanism.

\paragraph*{A slowest-class gating law for mixed workloads.}
In both systems the unit of parallelism in this regime is the lane: a lane's
slides are sequentially dependent, so no scheduler can split one lane across
threads. For a mix of query classes $c$ with $n_c$ lanes of per-slide cost
$w_c$, fixed-work speedup is therefore bounded by
\[
S(T) \;\le\; \frac{\sum_c n_c\, w_c}{\max_{\text{driver}} \sum_{\ell \in
\text{owned}} w_\ell},
\]
\[
\qquad\text{and for } T \ge n_{\max}:\quad
S_\infty \;=\; \frac{\sum_c n_c\, w_c}{w_{\max}},
\]
where $w_{\max}$ is the per-slide cost of the most expensive class and
$n_{\max}$ its lane count: once every expensive lane has a thread to itself,
the wall clock is one such lane running alone, and additional threads only
idle. The ceiling is a property of the \emph{cost ratio between query
classes}, not of thread count.

Both systems obey this law; they hit it at very different points because
their class-cost ratios differ. From the solo per-lane rates, TrieGS's
per-slide costs are $w_{\text{path}} \approx 3.4$\,ms, $w_{\text{star}}
\approx 28$\,ms, $w_{\text{spider}} \approx 37$\,ms, giving a predicted
\textsc{Mix3} ceiling of $\approx 14.7\times$; measured: $13.1\times$ at
$T{=}24$ (and $14.0\times$ in an independent longer-measurement run), i.e.\
TrieGS at 24 threads is already pressed against its own gating ceiling. For
Flink the spider costs $\approx 1{,}206$\,ms per slide against $21$\,ms for
the path --- the same $\sim$30$\times$ engine gap visible in the homogeneous
sweeps --- so the spider term dominates the sum and the predicted ceiling
collapses to $\approx 8\times(1+\epsilon)$ with $\epsilon \approx
0.05$--$0.1$; measured: a flat $7.57$--$7.68\times$ from $T{=}8$ through
$T{=}24$. Flink hits the identical wall $\sim$4$\times$ earlier in thread
count for one reason only: its per-slide cost on the expensive class is
$\sim$30$\times$ higher, which moves the mix's cost mass onto 8 unsplittable
lanes.

The interference control confirms the mechanism directly: inside
\textsc{Mix3} at $T{=}24$, \emph{every} lane --- path, spider, and star
alike --- advances at the same rate, which is $0.90\times$ the
\emph{spider's} solo rate and only $0.08\times$ the path's. The mix does not
average its classes; it is dragged to its slowest one.

\paragraph*{Ownership quantization and heterogeneous cores.}
Two second-order effects recur across every sweep and are worth naming
because either could be misread as an engine defect. First, at $T{=}16$ both
systems dip below their scaling trend on every mix (e.g.\ \textsc{Path3}:
TrieGS $10.5\times$, Flink $10.6\times$): 24 lanes over 16 threads leaves
eight threads owning two lanes each, so the ideal at that point is
$12\times$, not $16\times$ --- both systems sit at $\approx 0.88$ of the
quantized ideal, and the dip's presence \emph{in both systems at the same
point} is itself evidence that the harnesses impose identical ownership
structure. The spider mix's $T{=}16$ point is additionally run-sensitive
($11.2\times$ vs.\ $6.6\times$ across two committed runs of different
measurement lengths), consistent with quantization amplifying scheduling and
GC noise; we accordingly do not interpret $T{=}16$ levels, only the
$T{=}24$ endpoints and the curve shapes. Second, the efficiency cores help
throughput-oriented multi-query work (every mix improves from $T{=}16$ to
$T{=}24$) but \emph{hurt} compute-dense intra-query work: the spider's
morsel speedup regresses from $7.53\times$ at 16 workers to $5.97\times$ at
24, because a morsel barrier waits for its slowest worker and an efficiency
core is $\approx 0.55\times$ a performance core on this code.
Straggler-sensitive phases should not be scheduled onto heterogeneous cores
naively; this is a concrete input to the adaptive morsel-scheduling
direction in Section~\ref{sec:limitations}.

\paragraph*{Measurement discipline.}
Two protocol details materially affected validity. First, the checksum
work-identity control (bit-identical per-mix answer checksums across all
parallelism levels, asserted within every sweep) is what licenses reading
the curves as strong scaling at all. Second, warmup accounting is not a
formality on a JIT runtime: an early version of the sweep protocol re-paid
JIT compilation inside measured execution, roughly doubling
per-configuration execute time and inflating a full sweep severalfold while
contaminating early measurement periods; the committed protocol runs
explicit warmup periods excluded from timing. We flag this because scaling
studies on JVM-based engines that do not state their warmup protocol are
difficult to interpret.

\paragraph*{Insight.}
Mixed-workload scaling numbers are not comparable across systems without the
gating-law decomposition: a plateau at $T{=}8$ can indicate either a
scheduler defect or --- as here --- a per-class cost ratio that quantizes
away the remaining threads. The law also reads forward as a design
statement: the way to raise a mixed workload's ceiling is not more threads
but either (i) cheaper expensive classes --- exactly what the columnar
maintenance operator buys, moving the ceiling from $\approx 8\times$ to
$\approx 14.7\times$ on the same mix --- or (ii) intra-query parallelism
that breaks the lane atomicity assumption, which TrieGS's morsel path
provides for compute-dense plans (Section~\ref{sec:eval-anatomy}) and
Flink's per-key operator model structurally cannot.

\subsection{Fair Hybrid Scheduling: Reclaiming Gating-Induced Idleness}
\label{sec:eval-hybrid}

The gating law leaves a specific resource on the table. On the heterogeneous
mix at $C{=}24$, the fair pure-inter configuration runs at CPU utilization
0.59: once every expensive lane owns a thread, the remaining threads can only
idle, because a lane's slides are sequentially dependent. The natural design
response is a hybrid: keep the cheap lanes on inline driver threads, and pool
the expensive lanes --- each served by a lightweight feeder thread that
submits its slides' morsels into one shared work-stealing pool --- so the
threads the gating law would idle instead execute morsels of the pooled
lanes. The core budget is static, $D_{\mathrm{inline}} + W \le C$, where
$D_{\mathrm{inline}}$ is the inline driver count and $W$ the pool worker
count. This subsection measures whether that bridge between the two
parallelism granularities pays, under a fairness protocol strict enough that
the answer means something.

\paragraph*{Fairness methodology.}
A hybrid-vs-pure comparison is only as strong as its pure baseline, and the
obvious baseline is a strawman. Round-robin lane assignment (lane index
modulo driver count) ignores lane costs; on this mix at $C{=}16$ it reaches
$8.13\times$ where a cost-aware longest-processing-time (LPT) assignment
reaches $13.61\times$ --- the naive baseline gives up $\sim$40\% before the
hybrid is even started, and any lane-granular scaling study that round-robins
heterogeneous lanes is comparing against a strawman. All pure-inter baselines
below therefore use LPT. Threads are role-split into inline drivers, feeders,
and pool workers, with the static budget $D_{\mathrm{inline}} + W \le C$ and
feeders exempt as coordination overhead; the exemption is then audited by a
measured CPU-budget gate: total process CPU time must satisfy
$\mathrm{CPU}_{\mathrm{total}} \le 1.05 \cdot \mathrm{wall} \cdot C$
(\texttt{ThreadMXBean}), which is hardware-enforced at $C{=}24$ on this
24-core host. Every reported point passes the gate
(\texttt{budgetRespected}), and the feeders account for under 2\% of total
CPU, so the hybrid's gain is not hidden extra parallelism.

\paragraph*{Result: $+9.6\%$ over a fair baseline at equal core budget.}
All numbers are medians of three fresh-JVM repetitions on the heterogeneous
24-lane mix of Section~\ref{sec:eval-multiquery}. The fair pure-inter LPT
baseline at $C{=}24$ reaches $13.07\times$ at utilization 0.59. The best
hybrid pools only the eight spider lanes behind eight feeders, with
$D_{\mathrm{inline}}{=}9$ inline drivers owning the sixteen cheap lanes,
$W{=}15$ pool workers, and per-slide morsel fan-out 2: $14.33\times$ at
utilization 0.87, i.e.\ $+9.6\%$ over the fair baseline at an identical core
budget. The fan-out-4 variant medians $14.16\times$ with a single-repetition
maximum of $15.15\times$: the $\ge 15\times$ stretch target is touched once
but is not robust across repetitions, and we do not claim it.
Table~\ref{tab:hybrid-fair} summarizes the key configurations and controls.

\begin{table}[t]
\centering
\small
\caption{Fair hybrid scheduling on the heterogeneous 24-lane mix
(speedup vs.\ $T{=}1$; medians of 3 fresh-JVM repetitions; every point
passes the measured CPU-budget gate
$\mathrm{CPU}_{\mathrm{total}} \le 1.05 \cdot \mathrm{wall} \cdot C$).}
\label{tab:hybrid-fair}
\begin{tabular}{lrrr}
\toprule
Configuration & $C$ & Speedup & CPU util. \\
\midrule
Fair pure-inter (LPT) & 24 & 13.07x & 0.59 \\
Fair pure-inter (LPT) & 16 & 13.61x & 0.89 \\
Strawman pure-inter (round-robin) & 16 & 8.13x & 0.53 \\
Hybrid, spiders-only, fan-out 2 & 24 & \textbf{14.33x} & 0.87 \\
Hybrid, spiders-only, fan-out 4 & 24 & 14.16x & 0.88 \\
Hybrid falsification control & 8 & 3.88x & 0.59 \\
Homogeneous-spider control, hybrid & 24 & 12.6x & --- \\
Homogeneous-spider control, pure & 24 & 18.1x & --- \\
\bottomrule
\end{tabular}
\end{table}

\paragraph*{Mechanism: idle-capacity reclaim, not gate movement.}
The hybrid does not meaningfully accelerate the expensive class: the pooled
spider lanes' per-slide wall drops only from 42.7 to 38.8\,ms ($-9\%$), and
the star lanes --- now sharing nine inline drivers instead of owning more ---
rise from 31.8 to 35.3\,ms. The gating ceiling itself therefore barely moves;
the entire win is utilization, 0.59 to 0.87. This is exactly what the
structural-width ceiling of Section~\ref{sec:eval-anatomy} predicts: a pooled
lane's morsels run at width-law efficiency, $\approx 1.3$--$1.6\times$
effective at fine-grained delta sizes, so pooling cannot move the gate --- it
can only fill the idle cores with that modestly-parallel work. The reclaimed
throughput is packing, not acceleration.

\paragraph*{Controls.}
Three controls license the $+9.6\%$ as real rather than artifact. First, a
falsification control at $C{=}8$: with no gating-induced idleness to reclaim
(the core budget is scarcer than the expensive-lane count), the same hybrid
collapses to $3.88\times$ against a pure-inter $7.66\times$, with the spider
wall inflating to 141\,ms --- the win is regime-specific, as the mechanism
requires. Second, a homogeneous control of 24 identical spider lanes, where
utilization is already high and the gating law leaves nothing idle: pure
inter-query wins decisively, $18.1\times$ against the hybrid's $12.6\times$,
isolating the pool's coordination overhead. Third, a naive adaptive
controller that water-filled workers by static cost shares chose
$D_{\mathrm{inline}}{=}1, W{=}23$ and collapsed to $2.55\times$; it was
replaced by a measured calibrate-and-select controller that briefly runs the
candidate configurations and picks by observed throughput, which recovers the
pure-inter configuration when $C \ge$ lane count and the hybrid when
$C <$ lane count. Adaptivity here must be measured, not modeled.

\paragraph*{Latency under open-loop arrival.}
Fixed-work strong scaling measures capacity; a scheduler is also judged by
tail latency under arrival, and the two objectives need not agree. An
open-loop companion study drives the same 24 lanes with scheduled arrivals
(every lane one slide per period $P$, per-lane FIFO, latency measured as
completion minus scheduled arrival; 200 arrivals per lane, two fresh-JVM
repetitions) at two load points, $\approx 70\%$ and $\approx 93\%$ of the
measured $C{=}24$ capacity. At moderate load, pure-inter wins the tail
everywhere --- dedicated per-lane threads add no scheduling layers, and at
70\% utilization there is little queueing to absorb. Near saturation the
ordering inverts decisively: the pure-inter spider thread runs at
$\approx 0.83$ utilization and its per-lane queue tail explodes to a p99 of
$148$--$180$\,ms across the two repetitions, while the spiders-only hybrid
holds the spider p99 at $62$--$68$\,ms --- $2.2$--$2.9\times$ lower --- and
the aggregate p99 at $59$--$61$\,ms against $77$--$142$\,ms, at slightly
higher achieved throughput. The mechanism is again not acceleration: the
pooled spider's mean service time is unchanged ($\approx 37$\,ms, as the
width law requires when eight simultaneous arrivals share fifteen workers).
The pool buys \emph{statistical multiplexing} --- burst capacity shared
across the pooled lanes, so a slow slide no longer serializes behind its own
lane's dedicated thread. Notably, the configuration this study was designed
to vindicate is instead refuted: deep pooling (all sixteen expensive lanes
behind one 23-worker pool), whose halved per-slide wall under fixed work
(42.7 to 24.9\,ms) motivated the frontier hypothesis, is dominated at both
load points --- that service-time win was an artifact of round-robin pacing
and does not survive simultaneous open-loop arrivals, which congest the
shared pool (its spider service time \emph{rises} to $42$--$43$\,ms). The
capacity winner and the tail-latency winner near saturation are the same
configuration, and the crossover with pure-inter is load-dependent
($\approx 70\%$ utilization on this mix) --- which extends the
calibrate-and-select controller's remit naturally from a throughput
objective to a latency SLO.

\paragraph*{Insight.}
The results compose into a third law, the \emph{idle-capacity reclaim law}:
hybrid pooling pays if and only if the gating law leaves utilization below
one, and its gain is bounded by the width law, because the reclaimed cores
run the pooled lane's morsels at $\approx 1.3$--$1.6\times$ effective
parallelism. The reclaimable band --- between the gating ceiling
($\approx 14.7\times$ on this mix, Section~\ref{sec:eval-multiquery}) and
the heterogeneous hardware ideal ($\approx 20.4\times$) --- is therefore
only partially recoverable by any scheduler that cannot break lane
atomicity more efficiently than the structural width permits. A second,
independent finding falls out of the fair baseline itself: fair pure-inter
peaks at $C{=}16$ ($13.61\times$), \emph{above} its $C{=}24$ level
($13.07\times$) --- on this host the efficiency cores are a net negative
for pure inter-query execution of this mix, a placement effect consistent
with the straggler sensitivity of Section~\ref{sec:eval-multiquery}.

\subsection{NUMA-Node Affinity: Placement Is a First-Order Effect}
\label{sec:eval-intel}

The NUMA host exposes 112 logical CPUs, 4 sockets, 14 cores per socket, 2
threads per core, and 4 NUMA nodes. Each node exposes 28 logical CPUs. 

The sweep contains 660 JMH measurements: 5 placements, 12 worker counts, and 11
SNB-shaped sliding-window query benchmarks. The default scheduler placement
peaks at 8 workers with a geomean of 1{,}399 ops/s. CPU-node affinity changes
the curve: throughput jumps between 4 and 8 workers and then remains nearly flat
across larger worker counts. The best run is node 2 at 16 workers with a geomean
of 3{,}564 ops/s. This is 2.55x over the best default placement and 5.71x over
the default one-worker baseline.

\begin{table*}[t]
\centering
\small
\caption{Intel CPU-node affinity summary. Throughput is geomean ops/s over the
11 SNB-shaped sliding-window queries.}
\label{tab:intel-cpu-node-summary}
\begin{tabular}{lrrr}
\toprule
Placement & Best workers & Best geomean ops/s & Vs. default best \\
\midrule
Default scheduler placement & 8 & 1{,}399 & 1.00x \\
CPU-node 0 & 32 & 3{,}558 & 2.54x \\
CPU-node 1 & 16 & 3{,}548 & 2.54x \\
CPU-node 2 & 16 & 3{,}564 & 2.55x \\
CPU-node 3 & 64 & 3{,}539 & 2.53x \\
\bottomrule
\end{tabular}
\end{table*}

The four CPU-node best results are very consistent: mean 3{,}552 ops/s,
standard deviation 11 ops/s, and coefficient of variation about 0.3\%. The
effect is not localized to one favored socket. Every query improves under
CPU-node affinity; the per-query improvement range is 2.33x to 2.81x over the
best default placement for the same query.

\begin{table}[t]
\centering
\small
\caption{Intel CPU-affinity effect.}
\label{tab:intel-affinity-effect}
\begin{tabular}{lr}
\toprule
Quantity & Result \\
\midrule
Default best & 1{,}399 ops/s at 8 workers \\
Best pinned run & 3{,}564 ops/s at 16 workers on node 2 \\
Speedup over default best & 2.55x \\
Speedup over default 1 worker & 5.71x \\
Per-query improvement range & 2.33x--2.81x \\
CPU-node best-result CV & about 0.3\% \\
\bottomrule
\end{tabular}
\end{table}

\paragraph*{Insight.}
The NUMA result is the clearest systems insight in the current evaluation:
placement can dominate algorithmic tuning at this stage. Default scheduling
leaves substantial throughput on the table, while pinning workers to a CPU node
makes performance both higher and more stable. This supports the TrieGS's
morselized scheduling direction, especially prefix ownership and explicit
placement. 

\subsection{Homogeneous-Core Replication: The Three Laws Off UMA}
\label{sec:eval-intel-replication}

Every scaling number above comes from one 24-core UMA node with a
16-performance/8-efficiency split, and two of the three laws could in
principle be silicon artifacts: the structural-width ceiling
(Section~\ref{sec:eval-anatomy}) could reflect the UMA's memory system, and
the fair-hybrid finding that pure inter-query peaks at $C{=}16$ rather than
$C{=}24$ (Section~\ref{sec:eval-hybrid}) is explicitly attributed to the
efficiency cores. We therefore replicated the single-query, multi-query, and
hybrid experiments on the NUMA host of Section~\ref{sec:eval-intel} --- 4
sockets $\times$ 14 cores $\times$ 2 threads, all \emph{homogeneous} --- pinned
to one NUMA node (28 logical CPUs) with \texttt{numactl --cpunodebind --membind}
(true memory-local placement; the container here carries \texttt{CAP\_SYS\_NICE},
unlike Section~\ref{sec:eval-intel}), and to the whole machine.

\paragraph*{The structural-width law is hardware-independent.}
Single-query per-slide maintenance on 28 homogeneous cores plateaus almost
immediately: at fan-out 8 it reaches $1.14\times$ (peak at four workers, then
flat), and at fan-out 16 $1.39\times$ (peak at eight, then flat) --- if
anything a \emph{lower} ceiling than the Intel-NUMA's $\approx 1.6\times$, on faster
scaling silicon. The compute-density invariance reproduces exactly ($4\times$
denser joins, same curve), as does the per-slide allocation
(\texttt{gc.alloc.rate.norm} $14.06$ / $54.4$\,MB at fan-out 8 / 16, identical
to the Apple's UMA --- the maintenance is deterministic). A memory-bandwidth or
delta-size explanation predicts the two architectures separate; they do not.
The ceiling is structural parallel width, confirmed on a second, homogeneous
memory system.

\paragraph*{The $C{=}16$ anomaly was the efficiency cores.}
Multi-query \textsc{Mix3} strong scaling on the homogeneous node is
\emph{monotone} --- $1.86 / 3.27 / 5.34 / 8.04 / 9.67\times$ at
$1/2/4/7/14/28$ drivers --- with none of the Apple's UMA's peak-at-$C{=}16$,
dip-at-$C{=}24$ signature. That non-monotonicity was therefore a placement
artifact of the 16P+8E split, exactly as Section~\ref{sec:eval-hybrid}
conjectured, not an engine or gating effect. Across all four sockets the whole
machine reaches $28.8\times$ at 56 drivers, rolling off to $25.9\times$ at 112
(SMT plus cross-socket traffic); the gating and framework-grade-scaling stories
of Section~\ref{sec:eval-multiquery} hold on homogeneous cores.

\paragraph*{The reclaim law holds, but the pooling set is cost-ratio-dependent.}
The idle-capacity reclaim reproduces in kind but not in magnitude, and the
difference is instructive. The single fair-hybrid point (pool the spiders,
$D_{\mathrm{inline}}{=}10$, $W{=}18$) lifts utilization $0.55{\to}0.83$ but
throughput only $+0.4\%$ ($9.96\times$ vs.\ $9.92\times$), because on this host
the star class costs $108$\,ms per slide against the spider's $135$\,ms ---
nearly flat, versus the UMA's $29$ vs.\ $42$. Pooling only spiders then starves
the stars left on the inline drivers (star wall $108{\to}144$\,ms), cancelling
the gain. The same effect is sharper under open-loop arrival (load points
auto-calibrated to the node's measured $\approx 182$ slides/s capacity):
at $\approx 70\%$ load pure inter-query wins the tail (aggregate p99 $196$\,ms
vs.\ $211$/$213$), but near saturation ($\approx 93\%$) pure inter-query
\emph{saturates} and its spider p99 explodes to $1451$\,ms, while pooling
\emph{both} expensive classes (deep pooling) keeps every class bounded
(aggregate p99 $258$\,ms) and the spiders-only hybrid sits in between
($797$\,ms) precisely because it leaves the stars inline. So the reclaim law is
architecture-independent, but the \emph{set} of lanes worth pooling tracks the
class-cost spread: spiders-only when one class dominates (Apple Silicon, spider-and-star
when the expensive classes are comparable (Intel). This is a direct argument
for the measured calibrate-and-select controller of
Section~\ref{sec:eval-hybrid} over any fixed pooling heuristic.

\paragraph*{Insight.}
The replication does what a second architecture should: it promotes the
structural-width ceiling from an Apple's UMA observation to a hardware-independent
property, explains away the one non-monotone anomaly in the UMA data as an
efficiency-core placement effect, and stress-tests the reclaim law into a
sharper, cost-ratio-aware form. The two laws that are properties of the
workload (gating, width) transfer unchanged; the one that is a property of the
scheduler (reclaim) transfers in mechanism and generalizes in policy. 

\subsection{Operator Generality and the Hardware Envelope}
\label{sec:eval-operators}

The three laws were established on the streaming maintenance operator. Do they
describe that operator, or streaming operators in general? We test two more
operator classes --- a worst-case-optimal (WCOJ) triangle join and a full-RDFS
RETE reasoner --- under the same instrument as the multi-query study: $N$
independent single-threaded instances on $N$ threads, aggregate throughput,
efficiency $E(N)=T(N)/(N\,T(1))$. The new ingredient is a \emph{spin control}:
a pure-ALU benchmark (\texttt{Blackhole.consumeCPU}) co-run at every thread
count, so the hardware envelope can be divided out, $\hat{E}(N)=E_{\mathrm{op}}/
E_{\mathrm{spin}}$. Without it, ``the join scales to $23\times$'' and ``the ALU
scales to $86\times$'' at the same 112 threads are not comparable; with it, both
become a fraction of what the silicon delivers. We run on the Apple UMA (16P+8E, single
memory domain) and the homogeneous 4-socket Intel NUMA host, where the spin control is
essential because the raw curves span P/E, cross-socket, and SMT segments.

\paragraph*{Operator scaling is memory-bound, not compute-bound.}
The spin control scales near-ideally --- $19.7\times$ at 24 threads on the Apple
(97\% of the $\approx 20.4\times$ heterogeneous ideal), and $46\times$ at 56 and
$85.8\times$ at 112 on Intel. The operators do not come close. At full width the
WCOJ join reaches $\hat{E}=0.55$ (M2, 24T) and $0.41$--$0.27$ (Intel, 56--112T);
the reasoner reaches $\hat{E}=0.48$ (M2) and collapses to $0.15$ at 112T on
Intel. Independent operator instances share no algorithmic state --- they are
embarrassingly parallel by construction --- so the shortfall is entirely the
memory system: at scale these operators are bandwidth- and allocation-bound, and
the spin divisor makes that quantitative rather than asserted. This is the same
finding the intra-query anatomy reached for one operator
(Section~\ref{sec:eval-anatomy}), now shown to be a property of the operator
\emph{class}, not of TrieGS maintenance.

\paragraph*{Reasoning hits the memory wall before the join.}
The two operators are not bounded equally, and the ordering is consistent and
mechanistic. In a single memory-local domain --- the Apple M2 up to a few threads, and
Intel's memory-bound Region A (14 cores, \texttt{--membind}) --- the reasoner
actually scales \emph{better} than the join ($\hat{E}=0.86$ vs.\ $0.73$ at 14
Intel cores), because its per-arrival work amortizes framework overhead. But as
memory domains and contention grow it saturates first: on the M2's 16P+8E part
its normalized efficiency falls below the join's beyond eight threads
($0.48$ vs.\ $0.55$ at 24), and on Intel it \emph{peaks at 28 threads} --- two
sockets --- at $14.8\times$ and then \emph{declines} ($13.4\times$ at 42,
$12.6\times$ at 112), while the join keeps climbing to $19\times$ at 56 and
$23\times$ at 112. The reasoner's larger per-arrival working set and materialized
inference make it the first to pay cross-socket remote-memory cost; the join is
the more hardware-bound operator. The law is general (both are memory-bound); the
constant that matters is the operator's bytes-touched-per-unit-work.

\subsection{Windowing Restores Reasoning Scalability}
\label{sec:eval-windowed-reasoning}

The previous section measured the reasoner \emph{unbounded}: a full-RDFS RETE
network whose working memory grows without limit, which is the worst case for a
bandwidth-bound operator and the reason it saturates first. But streaming
reasoning is not unbounded --- it runs over a \emph{window}. If the memory wall is
set by working-set size, bounding the window should move it. We test this on a
real reasoning workload rather than a synthetic TBox: LUBM(1) (103{,}104 triples
from the Lehigh UBA generator) under the real VLog~\cite{DBLP:conf/semweb/CarralDGJKU19} \texttt{LUBM\_L} rule set (170
recursive Datalog rules), replayed as a windowed event stream (5\,s window, so the
per-stream working set holds $\approx$5{,}000 facts) through the native RETE path.
Because a single stream cannot be hash-partitioned across cores without breaking
cross-entity joins (e.g.\ $RP6(X)\!:\!-RP2(X),RP19(X,X_1),RP14(X_1)$ joins on $X$
and $X_1$), we run $N$ independent windowed streams on $N$ threads --- the
multi-tenant capacity question --- under the same aggregate-throughput,
spin-normalized instrument. A soundness check anchors correctness: on an unbounded
window the windowed stream and a batch pass converge to the identical
$121{,}512$-fact closure.

\begin{table}[t]
\centering
\small
\caption{Multi-core windowed streaming reasoning (LUBM(1)$\times$\texttt{LUBM\_L},
M2 16P+8E). Windowed reasoning tracks the pure-ALU spin envelope through the 16
performance cores; the unbounded reasoner of
Section~\ref{sec:eval-operators} reached only $\hat{E}=0.48$ at 24 threads.}
\label{tab:windowed-reasoning-scaling}
\begin{tabular}{rrrrr}
\toprule
Threads & Aggregate kev/s & Speedup & $E(N)$ & Spin $E(N)$ \\
\midrule
1  & 17.4  & 1.00x  & 1.00 & 1.00 \\
2  & 34.5  & 1.98x  & 0.99 & 0.96 \\
4  & 70.0  & 4.02x  & 1.01 & 0.94 \\
8  & 142.2 & 8.17x  & 1.02 & 0.94 \\
16 & 286.4 & 16.45x & 1.03 & 0.93 \\
24 & 316.1 & 18.16x & 0.76 & 0.82 \\
\bottomrule
\end{tabular}
\end{table}

\paragraph*{The window is what restores scaling.}
Windowed reasoning scales near-linearly to $16.45\times$ across the 16 performance
cores, with efficiency $E(N)\approx 1.0$ that meets or exceeds the pure-ALU spin
control at every point through 16 threads
(Table~\ref{tab:windowed-reasoning-scaling}); normalized against the spin
envelope this is $\hat{E}\approx 1.0$ --- the operator runs at what the silicon
delivers. The same engine and rule engine, measured \emph{unbounded} in
Section~\ref{sec:eval-operators}, reached only $\hat{E}=0.48$ at 24 threads and was
the first operator to hit the wall. The sole difference is the window: bounding it
caps each stream's working set at a few thousand facts, small enough to stay
cache-resident, which removes the bandwidth pressure that throttled the unbounded
reasoner. The $16\!\to\!24$ dip ($E$ from 1.03 to 0.76) is not a memory wall but
the M2's eight slower efficiency cores joining --- the spin control dips in
lockstep (0.93 to 0.82). Aggregate single-node capacity reaches $\approx$286\,k
windowed-reasoning events/s on the performance cores.

\paragraph*{Bounded latency, and a capacity that is a first-class number.}
Per-event windowed maintenance is cheap and stable --- $p_{50}=50\,\mu$s,
$p_{99}=78\,\mu$s, $p_{99}/p_{50}=1.6\times$ --- an order of magnitude below the
insert-only accumulation the window replaces (488\,$\mu$s / 1.74\,ms), again
because the working set is bounded rather than growing. Replaying the measured
service times against a fixed offered-arrival schedule through a single-server
queue (the open-loop model of Section~\ref{sec:eval-hybrid}) yields a textbook
saturation curve: end-to-end latency is flat and sub-100\,$\mu$s to
$\approx$16\,k\,events/s, lifts through a knee at 16--19\,k, and past a single
core's capacity of $\approx$17--19\,k\,events/s the queue diverges. This is the
result a batch reasoner (LTG/VLog) cannot report at all --- a latency distribution
and a sustainable throughput, not one materialization wall-clock. The finding
sharpens Section~\ref{sec:eval-operators}'s closing constant: the reasoner's
bytes-touched-per-unit-work is set by the window, not fixed by the operator, and
windowing --- the defining feature of stream processing --- is the mechanism that
holds it small. This is the reasoning analogue of TrieGS's bounded incremental
maintenance. 

\paragraph*{Partitioning one stream compounds the effect --- superlinearly.}
The scaling above runs $N$ \emph{independent} streams; the harder question is
whether \emph{one} stream's reasoning workload can be shared across cores. It
cannot be naively hash-partitioned: six of the 170 rules join across entities
and one is a recursive transitive closure, so object-side facts are needed on
foreign partitions. We partition by subject hash and \emph{replicate} exactly
the facts a static rule analysis proves necessary: the base support closure of
every foreign body atom, pruned by a data cover-check (on LUBM(1): five
\texttt{rdf:type} classes plus \texttt{subOrganizationOf} --- 1{,}881 facts,
replication fraction $f=1.82\%$, versus $18.9\%$ for the unpruned sound
closure). Replication bounds the ideal speedup by an Amdahl-on-data law,
$S(N)=1/((1-f)/N+f)$: $12.6\times$ at $N{=}16$ for the pruned set, $4.2\times$
for the sound set. Measured, with a parity oracle certifying at every width that
the deduplicated derivation set is \emph{identical} to a single unpartitioned
instance: the sound set lands on its law ($3.97\times$ at 16) and the measured
broadcast fan-out matches $1+f(N-1)$ to three decimals --- but the pruned set
reaches $\mathbf{31.7\times}$ at $N{=}16$ (588\,k\,events/s from an 18.5\,k
baseline), \emph{superlinear} and $2.5\times$ above its own law's ceiling. The
law is not wrong; its constant-cost premise is. Partitioning divides the
\emph{window state} along with the arrivals: each worker's window holds
$\approx 5{,}000/N$ facts, join cost falls with working-set size, and per-event
service drops from $50\,\mu$s to $13.8\,\mu$s at $N{=}16$. The measured curve is
arrival-division ($\le 12.6\times$) compounded with service-cheapening
($\approx 3.6\times$). One partitioned stream thereby out-throughputs the
sixteen independent streams of Table~\ref{tab:windowed-reasoning-scaling} on the
same silicon (588\,k vs.\ 286\,k\,events/s): independent streams each carry the
full window; partitions carry $\approx$300 facts each. The plateau at $N{=}24$
($31.8\times$) is the single dispatcher thread saturating at
$\approx 1.7\,\mu$s/event, not the operators. The general lesson joins the two
laws of this section: the window bounds the working set, and partitioning
divides it --- for state-superlinear operators like reasoning joins, that
division is worth more than the cores it costs.

\subsection{What the Results Mean for TrieGS}
\label{sec:eval-interpretation}


\paragraph*{Root-vector publication is ahead of LFNT performance.}
The correctness and parity results support the semantic architecture. The
relation-backend matrix does not yet support LFNT as a performance contribution.

\paragraph*{The workload story is now credible.}
The evaluation now includes nontrivial sliding-window dynamic joins and temporal
joins. This makes the evaluation more realistic than relation microbenchmarks,
but final generality still requires RiverBench and larger workload sweeps.

\paragraph*{Placement is part of the algorithmic story.}
The Intel CPU-node result shows that multicore behavior is not only about worker
count. Worker placement changes the geomean by more than the difference between
many algorithmic variants in the current prototype. The final system should
therefore expose placement-aware scheduling as an evaluated design choice.

\paragraph*{The scalability claim lives at inter-query granularity.}
The engine-path study caps fine-grained intra-epoch parallelism at
$\approx 1.6$--$1.9\times$ and, with the width-law probes of
Section~\ref{sec:eval-anatomy}, explains why; the multi-query study shows
the same engine scaling at $13.1$--$17.7\times$ on 24 threads across three
mixes, matching a work-equivalent Flink baseline's curve shape while
exceeding its absolute throughput by $5.8\times$ to $\sim$32$\times$. The
pre-registered $\ge 8\times$ target is met by increasing the processing
workload --- more queries, or denser plans ($7.5$--$7.8\times$ intra-query
at 16 workers) --- which is the honest form of the classic streaming-engine
scalability claim, and the slowest-class gating law bounds what any
lane-parallel system can do on heterogeneous mixes.

\paragraph*{Three laws compose into one ceiling.}
The scaling results reduce to three named, separately measured laws. The
\emph{slowest-class gating law} bounds lane-parallel scaling of a mixed
workload by the cost ratio between query classes
(Section~\ref{sec:eval-multiquery}); the \emph{structural width law} caps
fine-grained intra-query parallelism at the small fixed number of
independently updatable structures, regardless of delta size or compute
density (Section~\ref{sec:eval-anatomy}); and the \emph{idle-capacity
reclaim law} states that hybrid pooling converts gating-induced idleness
into throughput, with a gain bounded by the width law
(Section~\ref{sec:eval-hybrid}). Composed, they bound what any scheduler
can achieve on a mixed workload: the ceiling is
$\min(\text{hardware ideal},$ $ \text{gating ceiling} +  \text{reclaimable
idleness})$, where the reclaimable term is discounted by the width-law
efficiency of the pooled lanes. On this host and mix the terms are
$\approx 20.4\times$, $\approx 14.7\times$, and a measured reclaim to
$14.33\times$. The first two laws are workload-structural rather than
engine-specific --- Flink measurably obeys the gating law, and its per-key
operator model makes the width law moot by forcing parallelism~1 per query
--- so TrieGS's advantage is \emph{level} (per-slide cost, the
$5.8$--$32\times$ absolute gap) plus the scheduler freedom (morsels, hybrid
pooling) to buy a real but bounded slice of the remaining idleness.

\subsection{Remaining Gaps}
\label{sec:eval-gaps}

The current implementations provide a strong first-round evaluation but the following gaps remain.

\begin{itemize}[leftmargin=*]
  \item Run RiverBench workloads to test RDF burstiness, duplicate assertions,
  watermarks, and stream-window behavior beyond SNB-shaped local data.
  \item Extend the Flink matrix beyond the committed single-node keyed
  comparison (Section~\ref{sec:eval-multiquery}): distributed and
  durable-state (RocksDB + checkpointing) configurations matched against a
  persistent TrieGS backend, a Flink SQL/Table-API formulation to bound the
  comparison from the usability side, repeated process runs, and resource
  accounting.
  \item Complete true memory-bound NUMA experiments if memory binding can be
  enabled; otherwise keep the claim at CPU-affinity level.
  \item Collect JFR, allocation, GC, perf, cache-miss, branch-miss, and remote
  memory counters.
  \item Add memory-lifetime metrics: retained bytes, copied nodes/bytes, active
  versions, retired nodes, and pinned-snapshot retention.
  \item Complete the SWAG aggregate matrix and clarify where aggregate summaries
  replace full enumeration.
  \nop{
  \item \emph{Done (Section~\ref{sec:eval-intel-replication}).}
  Replicating the sweeps on the homogeneous-core Intel/NUMA host with memory-local
  pinning separated core-type from worker-count effects: the structural-width
  ceiling reproduced (single-query $\le 1.4\times$, compute-density-invariant,
  identical allocation), the $C{=}16$ pure-inter anomaly vanished into a monotone
  curve on homogeneous cores (confirming it was an efficiency-core placement
  effect), and the reclaim law generalized to a cost-ratio-aware pooling set.
  Remaining hardware work is deeper profiling (remote-memory counters, JFR
  attribution), not the placement confound itself.
  }
  \item Distribute the lanes: the inter-query design shards trivially (lanes
  share no state and driver partitioning is already an ownership function),
  so a multi-node run tests only the ingest and publication fabric, not the
  operator.
  \item Break the structural width law. The shard-parallel delta-extraction
  attack listed in earlier drafts has since been implemented, oracle-gated,
  measured net-negative, and reverted (Section~\ref{sec:eval-anatomy}); the
  one identified remaining lever is partitioning the output view via a
  chunked copy-on-write columnar redesign so aggregation width scales with
  data rather than plan structure 
  \item Add core-type-aware morsel scheduling: the spider's
  $7.53\times{\to}5.97\times$ regression from 16 to 24 workers is a
  straggler effect at morsel barriers on efficiency cores, directly testable
  against Table~\ref{tab:intra-query-shapes}.
  \item Fold the measured capacity--latency frontier into the narrative: the
  open-loop study (200 arrivals per lane, 70\%/93\% load points, two
  repetitions) finds that near saturation the fair capacity winner is
  \emph{also} the tail-latency winner --- the pooled spider class's p99 falls
  from 148--180\,ms (pure-inter, whose dedicated-thread queue tail explodes
  at $\approx$0.83 utilization) to 62--68\,ms under spiders-only pooling,
  with unchanged pooled service time ($\approx$37\,ms): the win is
  statistical multiplexing across the pooled lanes, exactly as the width law
  requires. The deep-pooling configuration whose halved fixed-work wall
  motivated the frontier hypothesis is dominated at both load points --- its
  service-time win was a pacing artifact that does not survive simultaneous
  open-loop arrivals --- and pure-inter's dedicated threads win the tail
  below $\approx$70\% utilization. The load-dependent crossover is reported
  in Section~\ref{sec:eval-hybrid}; what remains is replicating it beyond
  one host.
\end{itemize}

In summary, the current implemetations establish correctness-first progress, expose a
useful negative backend result, demonstrate that the dynamic workload suite is
now nontrivial, and identify CPU affinity as a major multicore factor. The final
paper should preserve these insights while reserving final performance claims
for the completed protocol.

\section{Related Work}
\label{sec:related-work}

\subsection{Dynamic Constant-Delay Enumeration}

Acyclic query evaluation originates with Yannakakis's algorithm~\cite{yannakakis1981algorithms},
and free-connexity characterizes an important tractable region for
constant-delay enumeration~\cite{bagan2007acyclic}. Dynamic Yannakakis builds
compact representations for free-connex acyclic queries under updates
~\cite{DBLP:conf/sigmod/IdrisUV17,idris2019efficient}; General Dynamic
Yannakakis extends the framework to richer joins~\cite{IdrisUVVL20}. For the
more restrictive q-hierarchical class, constant update time and constant-delay
enumeration are possible under the corresponding dynamic model
~\cite{berkholz2017answering}. \TrieGS does not enlarge these query classes. Its
core theorem is deliberately restricted to full acyclic queries, with projected
free-connex queries supported through an explicit answer view.

CROWN develops change propagation without joins and supports constant-delay
enumeration of both full results and deltas~\cite{WangHDY23}. It is the closest
query-processing predecessor and the primary algorithmic reference in our
evaluation. CROWN's contribution is a join-free propagation plan that avoids
large intermediate views; its prototype maps plans to Flink DataStream
operators. \TrieGS addresses a different concurrency boundary: several
versioned relations live in one shared-memory process, a long-running
enumerator pins a completed state, and update workers concurrently construct
the next state. The atomic root vector is the protocol that prevents fractured
reads across those relations. Our CQELS baseline is therefore labeled
CROWN-style until it passes plan, full-result, delta-result, and qualitative
trend validation against the published construction.

\subsection{Incremental View Maintenance}

DBToaster materializes higher-order deltas, while F-IVM combines incremental
maintenance with factorized representations~\cite{DBLP:journals/pvldb/AhmadKKN12,
nikolic2018incremental,nikolic2020f}. Specialized heavy/light methods maintain
cyclic patterns such as triangles~\cite{DBLP:journals/tods/KaraNNOZ20}. These
techniques are complementary candidates for cyclic bag modules. Their payload
and maintenance view of query evaluation motivates our use of multiplicity
payloads, exact signed deltas, and primary/secondary/tertiary index roles.

\subsection{Compressed Query-Result Representations}

Deep and Koutris study compressed representations of conjunctive-query results
under adorned access patterns, explicitly trading representation space against
access delay and total answer time~\cite{deep2018compressed}. This model is a
useful lens for TrieGS's maintained access structures. The core TrieGS theorem
chooses the fully indexed constant-delay point of the design space, while a
compressed TrieGS variant could trade memory for larger local access delay.
Their setting is static and primarily concerns compressed access to query
outputs; TrieGS focuses on maintaining multiset query state under streaming
updates and publishing versioned snapshots for concurrent enumeration.

\subsection{Sliding-Window Aggregation}

SWAG-style algorithms maintain aggregate summaries over sliding windows for
associative operators, including non-invertible aggregates, without scanning the
whole window on every change~\cite{tangwongsan2020daba}. Later work extends the
model to out-of-order and bulk window changes~\cite{tangwongsan2023bulk}. CQELS
already contains SWAG aggregate operators, and TrieGS treats them as
complementary. SWAG is appropriate for aggregate-only outputs or as a downstream
consumer of exact TrieGS answer deltas. It does not provide the full weighted
answer-support representation or query-level root-vector snapshot required by
Theorem~\ref{thm:core}.

\subsection{Continuous Subgraph Matching}

GraphFlow, TurboFlux, SymBi, CaLiG, TC-Match, and RapidFlow use multiway search,
candidate structures, or dynamic indexes to discover matches in evolving
graphs~\cite{DBLP:conf/sigmod/KankanamgeSMCS17,
DBLP:conf/sigmod/KimSHLHCSJ18,DBLP:journals/pvldb/MinPPGIH21,
DBLP:journals/pacmmod/YangZZY23,DBLP:conf/icde/MinJPGIH24,
DBLP:journals/pvldb/SunSHL22}. A common-framework study reports that index
maintenance dominates some workloads, constant-delay approaches can pay high
index-update cost, and no method wins uniformly~\cite{DBLP:journals/pvldb/SunSLH22}.
\TrieGS differs by maintaining a CDE representation with a precise
snapshot-consistency contract. The common framework remains a useful secondary
baseline only when matching, duplicate, and update semantics agree.

\subsection{Parallel Joins and Stream-Processing Hardware}

Morsel-driven parallelism schedules small input fragments to worker threads that
execute query pipelines until pipeline breakers while favoring NUMA-local input
and operator state~\cite{DBLP:conf/sigmod/LeisBK014}. \TrieGS borrows this
scheduling principle but applies it to a different unit of work: affected Gmas,
guard-prefix scans, projection updates, and view-maintenance tasks inside one
dynamic update epoch. The additional constraint is semantic rather than only
load balancing: a scheduler must preserve the view invariant and may publish a
root vector only after every dependency in the epoch has completed. Thus
morselized maintenance scheduling is complementary to, not a replacement for,
the root-vector snapshot protocol.

Worst-case-optimal and free-join techniques avoid harmful intermediate results
in multiway joins~\cite{DBLP:journals/pacmmod/WangWS23,DBLP:journals/tods/MhedhbiKS21}.
HoneyComb shows that parallel WCOJ must address skew, index contention, and
redundant work rather than merely partitioning the top variable
~\cite{wu2025honeycomb}. These systems optimize one-shot or delta join
computation. \TrieGS instead maintains long-lived views and uses parallelism for
propagation, publication, and concurrent snapshot traversal. WCOJ remains a
natural cyclic-bag implementation and recomputation baseline.

Hypertree-decomposition theory offers a parallel-evaluation story of its own:
bounded-width evaluation is LOGCFL-complete and hence highly parallelizable in
principle; DB-SHUNT contracts any join tree in a logarithmic number of parallel
shunt steps regardless of tree shape; and GYM runs generalized Yannakakis on BSP,
costing plans by communication volume and synchronization
rounds~\cite{DBLP:conf/pods/GottlobGLS16}. Our measurements position \TrieGS on
this map in three ways. First, the guard-normal plan is itself a
hypertree-shaped object (guards as covering hyperedges, interface projections as
$\chi$-labels, separators as join-tree connectors), so the structural width law
of Section~\ref{sec:eval-anatomy} acquires a theoretical name: a per-update
delta traverses one width-$k$ path of the decomposition, and the observed
$\approx$1.64$\times$ intra-query ceiling is that bounded fan-out made visible
--- parallelizing the sibling semijoins of a width-$\le$3 path cannot pay at
tens-of-microseconds service times. Second, the PRAM-style results count
steps where our spin-normalized instrument counts bytes: shunt contraction
trades logarithmic depth for $O(r^2)$ intermediates, which is the wrong
direction in the bandwidth-bound regime every unbounded operator measurement
exhibits (Section~\ref{sec:eval-operators}); the profitable moves we measured
--- windowing and stream partitioning
(Section~\ref{sec:eval-windowed-reasoning}) --- instead shrink $r$, the base of
the decomposition's $r^k$ intermediate bound, per worker. Tree contraction
remains a candidate for the one place a full-tree evaluation survives in a
streaming engine, batch recomputation at window close, gated by whether its
intermediates stay memory-domain-resident. Third, the replication law of our
partitioned-stream design is GYM's communication-versus-parallelism trade-off
specialized to streams, and choosing the partition variable that minimizes the
replicated fraction $f$ is a fractional-cover problem over the same
decomposition --- the natural path from the hand-derived broadcast set of
Section~\ref{sec:eval-windowed-reasoning} to a decomposition-driven partitioning
planner.

Zeuch et al. analyze how data representation, queue-mediated operator
communication, synchronization, code generation, and windowing affect
single-node stream-processing efficiency~\cite{analyzing-vldb-spe-multicores}.
Their results motivate our hardware-counter and queueing analysis, but they do
not establish that LFNT has better cache or NUMA behavior. TrieGS therefore
compares the same logical plan across object-key and integer-key LFNT, MVCC,
persistent-map, and lock-based backends before attributing any improvement to
its physical representation.

\subsection{Concurrent Tries, Persistent Trees, and Snapshots}

C-tries provide concurrent dictionary operations and efficient non-blocking
snapshots; cache-tries add an auxiliary cache and analyze constant-time exact-key
operations under their hashing assumptions~\cite{DBLP:conf/ppopp/ProkopecBBO12,
DBLP:conf/ppopp/Prokopec18}. The cache-trie compaction work supplies
relation-level removal and compaction mechanisms~\cite{DBLP:conf/europar/Prokopec18}.
\LFNT instantiates these ideas for dictionary-coded Gma keys and multiplicity
payloads. The imported results are intentionally relation-local: they do not
prove secondary/tertiary query-access delay, dynamic CDE maintenance, or
cross-relation snapshot consistency. Root-vector publication and the view
invariant provide that lifting argument.

Purely functional structures provide another point in the design space. Aspen
uses compressed functional trees so readers can acquire lightweight immutable
graph snapshots while writers publish later roots, but must manage the space
and locality costs of persistence~\cite{dhulipala2019aspen}. Our persistent-map
baseline and reclamation experiments test the analogous trade-off for
query-specific materialized state rather than the base graph alone.

\subsection{Streaming RDF and Graph Workloads}

LDBC SNB Interactive v2 combines graph queries with concurrent updates and deep
deletions, making it suitable for testing turnstile maintenance and realistic
social-graph schemas~\cite{puroja2023ldbc}. RiverBench provides curated RDF
streams and metadata for reproducible streaming experiments
~\cite{sowinski2023riverbench}. We use these sources alongside controlled
synthetic streams because no one workload independently controls affected-state
size, answer support, duplicate frequency, skew, dictionary locality, and reader
lifetime.

\subsection{Multi-Query Optimization}

Materialized-view selection and shared computation have a long history
~\cite{sellis1988multiple,roy2000efficient,mistry2001materialized}. \TrieGS
first shares exactly equivalent canonical subplans; frequent pattern mining is
an optional optimizer on top. The evaluation measures marginal cost per query,
because a shared view is valuable only when its saved logical maintenance
exceeds its additional indexes and coordination. Reductions in cache misses,
branch mispredictions, coherence traffic, or NUMA accesses are measured rather
than inferred from sharing.

\section{Limitations and Future Work}
\label{sec:limitations}

\paragraph*{Query scope.}
The no-output-materialization theorem covers fixed full acyclic positive
conjunctive queries. It does not establish the same guarantee for arbitrary
projection, cycles, negation, aggregation, recursion, injective subgraph
matching, or full SPARQL. Extension modules must state their own semantics and
costs.

\paragraph*{Projected weighted answers.}
For projected free-connex queries, the present construction explicitly
materializes the count-aware answer relation $\mathsf{Ans}_Q$. Corollary~\ref{cor:projected-materialized}
therefore gives snapshot correctness and constant scan delay, but its update and
space costs include the materialized output support. We do not claim a compact
factorized projected iterator with duplicate-free weighted output. A future
construction would need separate proofs of one-path-per-head-tuple
uniqueness and count factorization.

\paragraph*{Update complexity.}
Constant delay does not imply constant update time. One edge update may touch a
large affected state $A_t$, and lower-level lock freedom cannot eliminate that
semantic work. The update theorem is affected-state sensitive.

\paragraph*{Access-structure space.}
The core theorem chooses the fully indexed constant-delay point. It does not
prove that maintained view state $M_t$, index state $I_t$, or retained version
state $R_t$ are linear in the input for every workload. Compressed access
representations that trade memory for larger local delay are possible future
variants, but they are outside the current theorem.

\paragraph*{Dictionary encoding and layout.}
Injective dictionary IDs preserve RDF equality and therefore query semantics,
but they do not guarantee hardware locality. Arrival-order identifiers may be
poorly correlated with join access patterns, and locality-aware remapping must
respect pinned snapshots and avoid unsafe ID reuse. Any claim that integer-coded
or packed LFNT improves cache, TLB, or NUMA behavior must be supported by the
layout and backend ablations in Sections~\ref{sec:evaluation} and~\ref{sec:eval-gaps}.

\paragraph*{Cost model.}
The $\OQ(1)$ statements use the expected cache-trie hashing model and amortize
structural resize work as specified in Section~\ref{sec:preliminaries}. They are
not deterministic worst-case bounds. Counts are also assumed to fit in a
machine word unless arbitrary-precision cost is reported separately.

\paragraph*{Hardware locality.}
The cache-trie cost model does not bound hardware-cache or TLB misses. A
pointer-based versioned LFNT can enlarge the working set, increase dependent
loads, retain old pages, and suffer remote NUMA access. Conversely, immutable
published versions can reduce reader--writer cache-line invalidation and
coherence traffic. Neither effect follows from the logical theorem. Hardware
locality claims therefore require the JFR/perf and NUMA experiments still listed
as final-protocol gaps in Section~\ref{sec:eval-gaps}.

\paragraph*{SWAG aggregate semantics.}
CQELS SWAG state is an aggregate summary, not a full answer-support
representation. COUNT, SUM, and AVG can consume signed answer deltas directly,
but MIN and MAX under arbitrary answer retractions need additional support
structures unless the declared semantics is a FIFO answer-event window. SWAG
integration is therefore an aggregate extension and does not change
Theorem~\ref{thm:core}.

\paragraph*{Morselized scheduling.}
Morselized epoch scheduling refines the serializable maintenance semantics only
when the scheduler satisfies epoch isolation, dependency closure, exact-final
value, publication-barrier, and snapshot-immutability invariants. It does not
change the query class or the correctness theorem. If prefix coalescing,
NUMA-prefix ownership, or adaptive worker allocation fail to improve the COST or
latency--memory--freshness frontier, the reusable contribution remains
root-vector publication rather than the scheduler. A direct engine-path
measurement (Section~\ref{sec:eval-m2}) first found no multicore speedup because
per-epoch cost was dominated by sequential $O(\text{state})$ snapshot bookkeeping
and by per-tuple allocation; rebuilding the maintenance to be $O(\Delta)$ and
columnar end to end---dense-integer columnar views and indices, a sharded
delta accumulator, single-scan change derivation, a mutable occurrence-count
ledger in place of the persistent base-state copy, and overlapping the answer
decode with the view merges---removes that wall and yields a genuine intra-epoch
speedup of $\approx 1.9\times$ at eight workers together with a several-fold
single-thread gain. The residual ceiling is now the structural parallel width
of the maintained state --- the small fixed number of independently updatable
views and relations --- established by compute-density invariance and a
net-negative enumeration-parallelization attempt
(Section~\ref{sec:eval-anatomy}); the two parallel phases (enumeration and
view merge) each scale about $2\times$ on eight workers once the sequential
phases are small, and the pre-registered $\ge 8\times$ target is not reached
inside one epoch. We verified this bound is
structural rather than layout or serial overhead: four attempts to parallelize the
residual work (partitioning the largest view, range-partitioning the index seal,
batching the base-state application, and block-striping the enumeration) were each
measured net-negative at the per-epoch delta sizes here, and improving cache
locality by contiguous-block striping regresses throughput by sacrificing load
balance. The intra-epoch ceiling is shape-dependent, not universal:
compute-dense plans reach $7.5$--$7.8\times$ at the 16 performance cores
(Section~\ref{sec:eval-anatomy}), the pre-registered $\ge 8\times$ target
is met at inter-query granularity (Section~\ref{sec:eval-multiquery}), and a
fair hybrid scheduler reclaims a further $+9.6\%$ of gating-induced idleness
at equal core budget, bounded by the same width law
(Section~\ref{sec:eval-hybrid}). The
columnar snapshots that enable the allocation-lean maintenance also give up
cross-epoch structural sharing: a slow reader pinning an epoch now retains
full columnar view copies rather than shared trie structure, so the
version-retention cost discussed below is higher per pinned epoch than in the
pointer-sharing design.

\paragraph*{Baseline scope.}
The Flink comparison is one node, heap state backend, object reuse on, no
checkpointing, no network exchange: it measures ``engine-embedded lanes
vs.\ framework keyed streams on one host,'' and says nothing about
distributed, checkpointed, or RocksDB-backed Flink deployments, whose
per-slide costs and failure guarantees differ in kind. The incremental
operator (variant A) is hand-written on Flink's runtime because Flink ships
no incremental multi-way join; it is bracketed from below by the idiomatic
variant B (native sliding windows, full recompute), but a Flink SQL or
Table-API formulation was not measured. Inversely, TrieGS lanes pay no
serialization or mailbox costs; the absolute gaps include that structural
advantage by design.

\paragraph*{Hardware and workload-model scope.}
All scaling results in Sections~\ref{sec:eval-anatomy},
\ref{sec:eval-multiquery}, and~\ref{sec:eval-hybrid} come from one 24-core
host with heterogeneous
cores (16P+8E); efficiency-core throughput ($\approx 0.55\times$ P) both
lowers the achievable ideal to $\approx 20.4\times$ and confounds
worker-count effects with core-type effects, and the Intel/NUMA replication
of these specific sweeps is pending. The workload model is fixed slide
schedules (fixed work per configuration), which measures capacity, not
latency under open-loop arrival; a system can look identical on fixed-work
strong scaling and differ under bursty arrival with queueing --- the
deep-pooling configuration that halves the expensive lane's per-slide wall
at a 12\% throughput cost (Section~\ref{sec:eval-gaps}) is exactly such a
case. The
slowest-class gating law of Section~\ref{sec:eval-multiquery} also binds
TrieGS: on mixes with a higher class-cost spread than the measured one,
TrieGS's own ceiling drops accordingly unless intra-query morsels can split
the expensive lanes, which currently pays off only for compute-dense plans;
hybrid pooling reclaims the gating law's idle threads ($+9.6\%$ at equal
core budget) but is itself bounded by the structural width law
(Section~\ref{sec:eval-hybrid}).

\paragraph*{Freshness.}
An enumerator sees the vector published when it starts, not necessarily the
newest wall-clock state when it finishes. Applications that require latest-state
linearizability may need shorter enumerations, cancellation, or a different
consistency contract.

\paragraph*{Version-retention cost.}
A slow reader pins old roots and can delay reclamation. The asymptotic delay
guarantee does not bound retained memory. This risk is measured directly in the
evaluation.

\paragraph*{Architectural necessity.}
The logical theorem needs versioned indexes and atomic query-state publication,
but not necessarily LFNT. The first-round local matrix in Section~\ref{sec:evaluation}
does not support an LFNT performance claim: simpler sequential and MVCC-object
backends dominate the configured prototype series. LFNT remains an empirical
design hypothesis until integer-coded, packed, and NUMA-aware variants outperform
simpler alternatives on the relevant latency--memory--freshness frontier.

\paragraph*{Durability.}
The current correctness proof is in-memory. Crash consistency, exactly-once
replay, persistent expiration queues, and atomic checkpointing of stream offsets
with root vectors require an additional protocol and validation.

\paragraph*{Workload validity.}
The value of concurrent full enumeration depends on applications that need
stable result traversal while updates continue. Delta-only consumers or
aggregate-only SWAG consumers may be more economical for many workloads. The
paper must demonstrate, not assume, the target regime.

\paragraph*{Evaluation maturity.}
Section~\ref{sec:evaluation} now reports first-round local measurements: semantic
oracle checks, publication snapshots, bounded Flink parity, relation-backend
ablation, Apple M2 scaling, an engine-path morsel-scaling study (an
$O(\Delta)$, columnar maintenance rebuild that turns an initial no-speedup result
into an $\approx 1.9\times$ intra-epoch multicore speedup with a several-fold
single-thread gain, bounded thereafter by structural parallel width), a
plan-shape
anatomy showing intra-query speedup tracks compute per resident byte
($1.3\times$ to $7.8\times$ across shapes), a 24-lane multi-query
strong-scaling study against a work-matched, parity-gated Flink baseline
($13.1$--$17.7\times$ at 24 threads; $5.8$--$32\times$ absolute advantage),
and a CPU-budget-gated fair hybrid-scheduling study ($+9.6\%$ over an LPT
pure-inter baseline at equal core budget).
Warmup accounting matters on a JIT runtime: the committed sweep protocol runs
explicit warmup periods excluded from timing, after an early protocol that
re-paid JIT compilation inside measured execution inflated sweeps severalfold.

\section{Conclusion}
\label{sec:conclusion}

\TrieGS studies the boundary between dynamic constant-delay enumeration and
shared-memory concurrent execution. The revised proof separates the classical
CDE part from the systems part. A free-connex witness-subtree argument gives the
abstract constant-delay enumerator; a multiset view-maintenance argument restores
counting payloads after each update epoch; an access-structure layer makes the
space--delay choice explicit; and dictionary-coded LFNT/root-vector refinement
turns that abstract state into a physical snapshot safe for concurrent readers.

For fixed full acyclic conjunctive queries under multiset semantics, the proof
chain is now
\[
  G_t\rightarrow L_t\rightarrow C_t\rightarrow S_t\rightarrow Q(G_t),
\]
where $L_t$ is the maintained logical view state, $C_t$ is the fully indexed
access representation over stable dictionary identifiers, and $S_t$ is the
published versioned root-vector snapshot. Projected free-connex queries are
supported by an explicit count-aware answer view and therefore include its
materialization costs. Morselized scheduling, SWAG aggregate summaries, packed
integer layouts, and locality-aware dictionary policies are refinements or
extensions, not new query-class theorems.

The practical claim remains conditional. LFNT is valuable only if it improves
the latency--memory--freshness frontier over simpler object-key and integer-key
versioned indexes, and concurrent full enumeration is important only where
delta consumption or aggregate-only SWAG summaries are not adequate substitutes.
The CQELS-based evaluation protocol therefore treats correctness, CROWN and
backend comparisons, dictionary-layout attribution, workload need, multicore
scaling, reclamation cost, SWAG aggregation, and multi-query economics as
separate falsifiable questions. 

\section*{Acknowledgment}
This project was supported by the Deutsche Forschungsgemeinschaft, German Research Foundation under grant number
453130567 (COSMO).
\balance
\bibliographystyle{abbrv}
\bibliography{main,evaluation_refs}

@article{puroja2023ldbc,
  title        = {The {LDBC} Social Network Benchmark Interactive Workload v2: A Transactional Graph Query Benchmark with Deep Delete Operations},
  author       = {P{\"u}roja, David and Waudby, Jack and Boncz, Peter and Sz{\'a}rnyas, G{\'a}bor},
  journal      = {CoRR},
  volume       = {abs/2307.04820},
  year         = {2023},
  url          = {https://arxiv.org/abs/2307.04820},
  doi          = {10.48550/arXiv.2307.04820}
}

@article{sowinski2023riverbench,
  title        = {{RiverBench}: An Open {RDF} Streaming Benchmark Suite},
  author       = {Sowi{\'n}ski, Piotr and Ganzha, Maria and Paprzycki, Marcin},
  journal      = {CoRR},
  volume       = {abs/2305.06226},
  year         = {2023},
  url          = {https://arxiv.org/abs/2305.06226},
  doi          = {10.48550/arXiv.2305.06226}
}

@article{dhulipala2019aspen,
  title        = {Low-Latency Graph Streaming Using Compressed Purely-Functional Trees},
  author       = {Dhulipala, Laxman and Shun, Julian and Blelloch, Guy E.},
  journal      = {CoRR},
  volume       = {abs/1904.08380},
  year         = {2019},
  url          = {https://arxiv.org/abs/1904.08380},
  doi          = {10.48550/arXiv.1904.08380}
}

@article{wu2025honeycomb,
  title        = {{HoneyComb}: A Parallel Worst-Case Optimal Join on Multicores},
  author       = {Wu, Jiacheng and Suciu, Dan},
  journal      = {CoRR},
  volume       = {abs/2502.06715},
  year         = {2025},
  url          = {https://arxiv.org/abs/2502.06715},
  doi          = {10.48550/arXiv.2502.06715}
}

@article{analyzing-vldb-spe-multicores,
  title        = {Analyzing Efficient Stream Processing on Modern Hardware},
  author       = {Zeuch, Steffen and Del Monte, Bonaventura and Karimov, Jeyhun and Lutz, Clemens and Renz, Manuel and Traub, Jonas and Bre{\ss}, Sebastian and Rabl, Tilmann and Markl, Volker},
  journal      = {Proceedings of the VLDB Endowment},
  volume       = {12},
  number       = {5},
  pages        = {516--530},
  year         = {2019},
  doi          = {10.14778/3303753.3303758},
  url          = {https://www.vldb.org/pvldb/vol12/p516-zeuch.pdf}
}

@inproceedings{deep2018compressed,
  author       = {Deep, Shaleen and Koutris, Paraschos},
  title        = {Compressed Representations of Conjunctive Query Results},
  booktitle    = {Proceedings of the 37th ACM SIGMOD-SIGACT-SIGAI Symposium on Principles of Database Systems},
  series       = {PODS '18},
  pages        = {307--322},
  year         = {2018},
  publisher    = {Association for Computing Machinery},
  doi          = {10.1145/3196959.3196979}
}

@article{tangwongsan2020daba,
  author       = {Tangwongsan, Kanat and Hirzel, Martin and Schneider, Scott},
  title        = {In-Order Sliding-Window Aggregation in Worst-Case Constant Time},
  journal      = {CoRR},
  volume       = {abs/2009.13768},
  year         = {2020},
  url          = {https://arxiv.org/abs/2009.13768},
  doi          = {10.48550/arXiv.2009.13768}
}

@article{tangwongsan2023bulk,
  author       = {Tangwongsan, Kanat and Hirzel, Martin and Schneider, Scott},
  title        = {Out-of-Order Sliding-Window Aggregation with Efficient Bulk Evictions and Insertions},
  journal      = {CoRR},
  volume       = {abs/2307.11210},
  year         = {2023},
  url          = {https://arxiv.org/abs/2307.11210},
  doi          = {10.48550/arXiv.2307.11210}
}

@inproceedings{DBLP:conf/semweb/CarralDGJKU19,
  author       = {David Carral and
                  Irina Dragoste and
                  Larry Gonz{\'{a}}lez and
                  Ceriel J. H. Jacobs and
                  Markus Kr{\"{o}}tzsch and
                  Jacopo Urbani},
  editor       = {Chiara Ghidini and
                  Olaf Hartig and
                  Maria Maleshkova and
                  Vojtech Sv{\'{a}}tek and
                  Isabel F. Cruz and
                  Aidan Hogan and
                  Jie Song and
                  Maxime Lefran{\c{c}}ois and
                  Fabien Gandon},
  title        = {VLog: {A} Rule Engine for Knowledge Graphs},
  booktitle    = {The Semantic Web - {ISWC} 2019 - 18th International Semantic Web Conference,
                  Auckland, New Zealand, October 26-30, 2019, Proceedings, Part {II}},
  series       = {Lecture Notes in Computer Science},
  volume       = {11779},
  pages        = {19--35},
  publisher    = {Springer},
  year         = {2019},
  url          = {https://doi.org/10.1007/978-3-030-30796-7\_2},
  doi          = {10.1007/978-3-030-30796-7\_2},
  bibsource    = {dblp computer science bibliography, https://dblp.org}
}

@inproceedings{Amdahl:1967:VSP:1465482.1465560,
  acmid = {1465560},
  address = {New York, NY, USA},
  author = {Amdahl, Gene M.},
  booktitle = {Proceedings of the April 18-20, 1967, spring joint computer conference},
  doi = {10.1145/1465482.1465560},
  location = {Atlantic City, New Jersey},
  numpages = {3},
  pages = {483--485},
  publisher = {ACM},
  series = {AFIPS '67 (Spring)},
  title = {Validity of the single processor approach to achieving large scale computing capabilities},
  year = 1967
}

@inproceedings{DBLP:conf/icdm/YanH02,
	author = {Xifeng Yan and Jiawei Han},
	bibsource = {dblp computer science bibliography, https://dblp.org},
	booktitle = {Proceedings of the 2002 {IEEE} International Conference on Data Mining {(ICDM} 2002), 9-12 December 2002, Maebashi City, Japan},
	pages = {721--724},
	publisher = {{IEEE} Computer Society},
	title = {gSpan: Graph-Based Substructure Pattern Mining},
	year = {2002}}

@inproceedings{yannakakis1981algorithms,
  title={Algorithms for acyclic database schemes},
  author={Yannakakis, Mihalis},
  booktitle={VLDB},
  volume={81},
  pages={82--94},
  year={1981}
}

@inproceedings{Pacaci:2020,
author = {Pacaci, Anil and Bonifati, Angela and \"{O}zsu, M. Tamer},
title = {Regular Path Query Evaluation on Streaming Graphs},
year = {2020},
isbn = {9781450367356},
publisher = {Association for Computing Machinery},
address = {New York, NY, USA},
url = {https://doi.org/10.1145/3318464.3389733},
doi = {10.1145/3318464.3389733},
booktitle = {Proceedings of the 2020 ACM SIGMOD International Conference on Management of Data},
pages = {1415–1430},
numpages = {16},
location = {Portland, OR, USA},
series = {SIGMOD '20}
}

@inproceedings{DBLP:conf/icde/PacaciBO22,
  author       = {Anil Pacaci and
                  Angela Bonifati and
                  M. Tamer {\"{O}}zsu},
  title        = {Evaluating Complex Queries on Streaming Graphs},
  booktitle    = {38th {IEEE} International Conference on Data Engineering, {ICDE} 2022,
                  Kuala Lumpur, Malaysia, May 9-12, 2022},
  pages        = {272--285},
  publisher    = {{IEEE}},
  year         = {2022},
  url          = {https://doi.org/10.1109/ICDE53745.2022.00025},
  doi          = {10.1109/ICDE53745.2022.00025},
  bibsource    = {dblp computer science bibliography, https://dblp.org}
}

@article{DBLP:journals/pvldb/AhmadKKN12,
  author    = {Yanif Ahmad and
               Oliver Kennedy and
               Christoph Koch and
               Milos Nikolic},
  title     = {DBToaster: Higher-order Delta Processing for Dynamic, Frequently Fresh
               Views},
  journal   = {Proc. {VLDB} Endow.},
  volume    = {5},
  number    = {10},
  pages     = {968--979},
  year      = {2012},
  url       = {http://vldb.org/pvldb/vol5/p968\_yanifahmad\_vldb2012.pdf},
  doi       = {10.14778/2336664.2336670},
  bibsource = {dblp computer science bibliography, https://dblp.org}
}

@inproceedings{berkholz2017answering,
  title={Answering conjunctive queries under updates},
  author={Berkholz, Christoph and Keppeler, Jens and Schweikardt, Nicole},
  booktitle={proceedings of the 36th ACM SIGMOD-SIGACT-SIGAI symposium on Principles of database systems},
  pages={303--318},
  year={2017}
}

@inproceedings{DBLP:conf/pods/GottlobGLS16,
  author    = {Georg Gottlob and
               Gianluigi Greco and
               Nicola Leone and
               Francesco Scarcello},
  editor    = {Tova Milo and
               Wang{-}Chiew Tan},
  title     = {Hypertree Decompositions: Questions and Answers},
  booktitle = {Proceedings of the 35th {ACM} {SIGMOD-SIGACT-SIGAI} Symposium on Principles
               of Database Systems, {PODS} 2016, San Francisco, CA, USA, June 26
               - July 01, 2016},
  pages     = {57--74},
  publisher = {{ACM}},
  year      = {2016},
  url       = {https://doi.org/10.1145/2902251.2902309},
  doi       = {10.1145/2902251.2902309},
  bibsource = {dblp computer science bibliography, https://dblp.org}
}

@article{DBLP:journals/tods/MhedhbiKS21,
  author    = {Amine Mhedhbi and
               Chathura Kankanamge and
               Semih Salihoglu},
  title     = {Optimizing One-time and Continuous Subgraph Queries using Worst-case
               Optimal Joins},
  journal   = {{ACM} Trans. Database Syst.},
  volume    = {46},
  number    = {2},
  pages     = {6:1--6:45},
  year      = {2021},
  url       = {https://doi.org/10.1145/3446980},
  doi       = {10.1145/3446980},
  bibsource = {dblp computer science bibliography, https://dblp.org}
}

@article{DBLP:journals/pvldb/SunSLH22,
  author       = {Xibo Sun and
                  Shixuan Sun and
                  Qiong Luo and
                  Bingsheng He},
  title        = {An In-Depth Study of Continuous Subgraph Matching},
  journal      = {Proc. {VLDB} Endow.},
  volume       = {15},
  number       = {7},
  pages        = {1403--1416},
  year         = {2022},
  url          = {https://www.vldb.org/pvldb/vol15/p1403-sun.pdf},
  bibsource    = {dblp computer science bibliography, https://dblp.org}
}

@article{DBLP:journals/tods/KaraNNOZ20,
  author    = {Ahmet Kara and
               Hung Q. Ngo and
               Milos Nikolic and
               Dan Olteanu and
               Haozhe Zhang},
  title     = {Maintaining Triangle Queries under Updates},
  journal   = {{ACM} Trans. Database Syst.},
  volume    = {45},
  number    = {3},
  pages     = {11:1--11:46},
  year      = {2020},
  url       = {https://doi.org/10.1145/3396375},
  doi       = {10.1145/3396375},
  bibsource = {dblp computer science bibliography, https://dblp.org}
}

@article{IdrisUVVL20,
  author    = {Muhammad Idris and
               Mart{\'{\i}}n Ugarte and
               Stijn Vansummeren and
               Hannes Voigt and
               Wolfgang Lehner},
  title     = {General dynamic Yannakakis: conjunctive queries with theta joins under
               updates},
  journal   = {{VLDB} J.},
  volume    = {29},
  number    = {2-3},
  pages     = {619--653},
  year      = {2020},
  url       = {https://doi.org/10.1007/s00778-019-00590-9},
  doi       = {10.1007/s00778-019-00590-9},
  bibsource = {dblp computer science bibliography, https://dblp.org}
}

@inproceedings{DBLP:conf/sigmod/IdrisUV17,
  author    = {Muhammad Idris and
               Mart{\'{\i}}n Ugarte and
               Stijn Vansummeren},
  editor    = {Semih Salihoglu and
               Wenchao Zhou and
               Rada Chirkova and
               Jun Yang and
               Dan Suciu},
  title     = {The Dynamic Yannakakis Algorithm: Compact and Efficient Query Processing
               Under Updates},
  booktitle = {Proceedings of the 2017 {ACM} International Conference on Management
               of Data, {SIGMOD} Conference 2017, Chicago, IL, USA, May 14-19, 2017},
  pages     = {1259--1274},
  publisher = {{ACM}},
  year      = {2017},
  url       = {https://doi.org/10.1145/3035918.3064027},
  doi       = {10.1145/3035918.3064027},
  bibsource = {dblp computer science bibliography, https://dblp.org}
}

@article{DBLP:journals/pvldb/SunSHL22,
  author    = {Shixuan Sun and
               Xibo Sun and
               Bingsheng He and
               Qiong Luo},
  title     = {RapidFlow: An Efficient Approach to Continuous Subgraph Matching},
  journal   = {Proc. {VLDB} Endow.},
  volume    = {15},
  number    = {11},
  pages     = {2415--2427},
  year      = {2022},
  url       = {https://www.vldb.org/pvldb/vol15/p2415-sun.pdf},
  bibsource = {dblp computer science bibliography, https://dblp.org}
}

@inproceedings{nikolic2020f,
  title={F-IVM: learning over fast-evolving relational data},
  author={Nikolic, Milos and Zhang, Haozhe and Kara, Ahmet and Olteanu, Dan},
  booktitle={Proceedings of the 2020 ACM SIGMOD International Conference on Management of Data},
  pages={2773--2776},
  year={2020}
}

@inproceedings{nikolic2018incremental,
  title={Incremental view maintenance with triple lock factorization benefits},
  author={Nikolic, Milos and Olteanu, Dan},
  booktitle={Proceedings of the 2018 International Conference on Management of Data},
  pages={365--380},
  year={2018}
}

@article{WangHDY23,
  author       = {Qichen Wang and
                  Xiao Hu and
                  Binyang Dai and
                  Ke Yi},
  title        = {Change Propagation Without Joins},
  journal      = {Proc. {VLDB} Endow.},
  volume       = {16},
  number       = {5},
  pages        = {1046--1058},
  year         = {2023},
  url          = {https://www.vldb.org/pvldb/vol16/p1046-hu.pdf},
  bibsource    = {dblp computer science bibliography, https://dblp.org}
}

@inproceedings{DBLP:conf/ppopp/Prokopec18,
  author       = {Aleksandar Prokopec},
  editor       = {Andreas Krall and
                  Thomas R. Gross},
  title        = {Cache-tries: concurrent lock-free hash tries with constant-time operations},
  booktitle    = {Proceedings of the 23rd {ACM} {SIGPLAN} Symposium on Principles and
                  Practice of Parallel Programming, PPoPP 2018, Vienna, Austria, February
                  24-28, 2018},
  pages        = {137--151},
  publisher    = {{ACM}},
  year         = {2018},
  url          = {https://doi.org/10.1145/3178487.3178498},
  doi          = {10.1145/3178487.3178498},
  bibsource    = {dblp computer science bibliography, https://dblp.org}
}

@inproceedings{DBLP:conf/europar/Prokopec18,
  author       = {Aleksandar Prokopec},
  editor       = {Marco Aldinucci and
                  Luca Padovani and
                  Massimo Torquati},
  title        = {Efficient Lock-Free Removing and Compaction for the Cache-Trie Data
                  Structure},
  booktitle    = {Euro-Par 2018: Parallel Processing - 24th International Conference
                  on Parallel and Distributed Computing, Turin, Italy, August 27-31,
                  2018, Proceedings},
  series       = {Lecture Notes in Computer Science},
  volume       = {11014},
  pages        = {575--589},
  publisher    = {Springer},
  year         = {2018},
  url          = {https://doi.org/10.1007/978-3-319-96983-1\_41},
  doi          = {10.1007/978-3-319-96983-1\_41},
  bibsource    = {dblp computer science bibliography, https://dblp.org}
}

@inproceedings{DBLP:conf/ppopp/ProkopecBBO12,
  author       = {Aleksandar Prokopec and
                  Nathan Grasso Bronson and
                  Phil Bagwell and
                  Martin Odersky},
  editor       = {J. Ramanujam and
                  P. Sadayappan},
  title        = {Concurrent tries with efficient non-blocking snapshots},
  booktitle    = {Proceedings of the 17th {ACM} {SIGPLAN} Symposium on Principles and
                  Practice of Parallel Programming, {PPOPP} 2012, New Orleans, LA, USA,
                  February 25-29, 2012},
  pages        = {151--160},
  publisher    = {{ACM}},
  year         = {2012},
  url          = {https://doi.org/10.1145/2145816.2145836},
  doi          = {10.1145/2145816.2145836},
  bibsource    = {dblp computer science bibliography, https://dblp.org}
}

@article{sellis1988multiple,
	author = {Sellis, Timos K},
	journal = {ACM Transactions on Database Systems (TODS)},
	number = {1},
	pages = {23--52},
	publisher = {ACM New York, NY, USA},
	title = {Multiple-query optimization},
	volume = {13},
	year = {1988}}

@inproceedings{roy2000efficient,
	author = {Roy, Prasan and Seshadri, Srinivasan and Sudarshan, S and Bhobe, Siddhesh},
	booktitle = {Proceedings of the 2000 ACM SIGMOD international conference on Management of data},
	pages = {249--260},
	title = {Efficient and extensible algorithms for multi query optimization},
	year = {2000}}

@inproceedings{mistry2001materialized,
	author = {Mistry, Hoshi and Roy, Prasan and Sudarshan, S and Ramamritham, Krithi},
	booktitle = {ACM SIGMOD Record},
	organization = {ACM},
	pages = {307--318},
	title = {Materialized view selection and maintenance using multi-query optimization},
	volume = {30},
	year = {2001}}

@article{idris2019efficient,
  title={Efficient query processing for dynamically changing datasets},
  author={Idris, Muhammad and Ugarte, Mart{\'\i}n and Vansummeren, Stijn and Voigt, Hannes and Lehner, Wolfgang},
  journal={ACM SIGMOD Record},
  volume={48},
  number={1},
  pages={33--40},
  year={2019},
  publisher={ACM New York, NY, USA}
}

@inproceedings{bagan2007acyclic,
  title={On acyclic conjunctive queries and constant delay enumeration},
  author={Bagan, Guillaume and Durand, Arnaud and Grandjean, Etienne},
  booktitle={International Workshop on Computer Science Logic},
  pages={208--222},
  year={2007},
  organization={Springer}
}

@article{DBLP:journals/pacmmod/WangWS23,
  author       = {Yisu Remy Wang and
                  Max Willsey and
                  Dan Suciu},
  title        = {Free Join: Unifying Worst-Case Optimal and Traditional Joins},
  journal      = {Proc. {ACM} Manag. Data},
  volume       = {1},
  number       = {2},
  pages        = {150:1--150:23},
  year         = {2023},
  url          = {https://doi.org/10.1145/3589295},
  doi          = {10.1145/3589295},
  bibsource    = {dblp computer science bibliography, https://dblp.org}
}

@inproceedings{DBLP:conf/sigmod/KankanamgeSMCS17,
  author       = {Chathura Kankanamge and
                  Siddhartha Sahu and
                  Amine Mhedhbi and
                  Jeremy Chen and
                  Semih Salihoglu},
  editor       = {Semih Salihoglu and
                  Wenchao Zhou and
                  Rada Chirkova and
                  Jun Yang and
                  Dan Suciu},
  title        = {Graphflow: An Active Graph Database},
  booktitle    = {Proceedings of the 2017 {ACM} International Conference on Management
                  of Data, {SIGMOD} Conference 2017, Chicago, IL, USA, May 14-19, 2017},
  pages        = {1695--1698},
  publisher    = {{ACM}},
  year         = {2017},
  url          = {https://doi.org/10.1145/3035918.3056445},
  doi          = {10.1145/3035918.3056445},
  bibsource    = {dblp computer science bibliography, https://dblp.org}
}

@inproceedings{DBLP:conf/sigmod/KimSHLHCSJ18,
  author       = {Kyoungmin Kim and
                  In Seo and
                  Wook{-}Shin Han and
                  Jeong{-}Hoon Lee and
                  Sungpack Hong and
                  Hassan Chafi and
                  Hyungyu Shin and
                  Geonhwa Jeong},
  editor       = {Gautam Das and
                  Christopher M. Jermaine and
                  Philip A. Bernstein},
  title        = {TurboFlux: {A} Fast Continuous Subgraph Matching System for Streaming
                  Graph Data},
  booktitle    = {Proceedings of the 2018 International Conference on Management of
                  Data, {SIGMOD} Conference 2018, Houston, TX, USA, June 10-15, 2018},
  pages        = {411--426},
  publisher    = {{ACM}},
  year         = {2018},
  url          = {https://doi.org/10.1145/3183713.3196917},
  doi          = {10.1145/3183713.3196917},
  bibsource    = {dblp computer science bibliography, https://dblp.org}
}

@inproceedings{DBLP:conf/sigmod/LeisBK014,
  author       = {Viktor Leis and
                  Peter A. Boncz and
                  Alfons Kemper and
                  Thomas Neumann},
  editor       = {Curtis E. Dyreson and
                  Feifei Li and
                  M. Tamer {\"{O}}zsu},
  title        = {Morsel-driven parallelism: a NUMA-aware query evaluation framework
                  for the many-core age},
  booktitle    = {International Conference on Management of Data, {SIGMOD} 2014, Snowbird,
                  UT, USA, June 22-27, 2014},
  pages        = {743--754},
  publisher    = {{ACM}},
  year         = {2014},
  url          = {https://doi.org/10.1145/2588555.2610507},
  doi          = {10.1145/2588555.2610507},
  bibsource    = {dblp computer science bibliography, https://dblp.org}
}

@article{DBLP:journals/pvldb/MinPPGIH21,
  author       = {Seunghwan Min and
                  Sung Gwan Park and
                  Kunsoo Park and
                  Dora Giammarresi and
                  Giuseppe F. Italiano and
                  Wook{-}Shin Han},
  title        = {Symmetric Continuous Subgraph Matching with Bidirectional Dynamic
                  Programming},
  journal      = {Proc. {VLDB} Endow.},
  volume       = {14},
  number       = {8},
  pages        = {1298--1310},
  year         = {2021}
}

@inproceedings{DBLP:conf/icde/MinJPGIH24,
  author       = {Seunghwan Min and
                  Jihoon Jang and
                  Kunsoo Park and
                  Dora Giammarresi and
                  Giuseppe F. Italiano and
                  Wook{-}Shin Han},
  title        = {Time-Constrained Continuous Subgraph Matching Using Temporal Information
                  for Filtering and Backtracking},
  booktitle    = {{ICDE}},
  pages        = {3257--3269},
  publisher    = {{IEEE}},
  year         = {2024}
}

@article{DBLP:journals/pacmmod/YangZZY23,
  author       = {Rongjian Yang and
                  Zhijie Zhang and
                  Weiguo Zheng and
                  Jeffrey Xu Yu},
  title        = {Fast Continuous Subgraph Matching over Streaming Graphs via Backtracking
                  Reduction},
  journal      = {Proc. {ACM} Manag. Data},
  volume       = {1},
  number       = {1},
  pages        = {15:1--15:26},
  year         = {2023}
}

@article{DBLP:journals/pvldb/ZeuchBRMKLRTM19,
  author       = {Steffen Zeuch and
                  Sebastian Bre{\ss} and
                  Tilmann Rabl and
                  Bonaventura Del Monte and
                  Jeyhun Karimov and
                  Clemens Lutz and
                  Manuel Renz and
                  Jonas Traub and
                  Volker Markl},
  title        = {Analyzing Efficient Stream Processing on Modern Hardware},
  journal      = {Proc. {VLDB} Endow.},
  volume       = {12},
  number       = {5},
  pages        = {516--530},
  year         = {2019},
  url          = {http://www.vldb.org/pvldb/vol12/p516-zeuch.pdf},
  doi          = {10.14778/3303753.3303758},
  bibsource    = {dblp computer science bibliography, https://dblp.org}
}

\end{document}